\documentclass[mnsc,nonblindrev]{informs3} 

\OneAndAHalfSpacedXI % current default line spacing
\usepackage[T1]{fontenc}
\usepackage{textcomp}

\TheoremsNumberedThrough     % Preferred (Theorem 1, Lemma 1, Theorem 2)
\ECRepeatTheorems

\EquationsNumberedThrough    % Default: (1), (2), ...
\MANUSCRIPTNO{}

\usepackage{natbib}
 \bibpunct[, ]{(}{)}{,}{a}{}{,}%
 \def\bibfont{\small}%

\usepackage{lipsum}

\newcommand\blfootnote[1]{%
  \begingroup
  \renewcommand\thefootnote{}\footnote{#1}%
  \addtocounter{footnote}{-1}%
  \endgroup
}

\usepackage{xcolor}
\definecolor{tencent_blue}{RGB}{0, 82, 217}
\definecolor{tencent_orange}{RGB}{238, 126, 71}
\definecolor{redorange}{RGB}{255, 68, 51}
\usepackage{color-edits}
\addauthor[Ruohan]{rz}{redorange}
\addauthor[Yuchen]{yh}{cyan}
\addauthor[Zhenling]{zj}{purple}
\addauthor[Shichao]{sh}{tencent_orange}
\usepackage{thm-restate}
\usepackage[colorlinks=true, citecolor=blue, linkcolor = red]{hyperref}

\usepackage{subcaption,arydshln,enumerate}
\usepackage[ruled,vlined]{algorithm2e}
\usepackage{anyfontsize}
\usepackage{esvect}
\usepackage{tikz}
\usepackage{pgfplots}
\usepackage{booktabs}
\usepackage{tabularx}
\usepgfplotslibrary{fillbetween}
\usetikzlibrary{patterns, intersections}

\newcommand{\bbP}{\mathbb{P}}

\newcommand{\bbR}{\mathbb{R}}
\newcommand{\bbE}{\mathbb{E}}

\newcommand{\bfw}{\mathbf{w}}

\newcommand{\calF}{\mathcal{F}}

\newcommand{\calW}{\mathcal{W}}

\newcommand{\calB}{\mathcal{B}}

\newcommand{\calN}{\mathcal{N}}

\newcommand{\calC}{\mathcal{C}}

\newcommand{\zero}{\mathbf{0}}
\newcommand{\diag}{\mbox{diag}}
\newcommand{\one}{\mathbf{1}}
\newcommand{\bbI}{\mathbb{I}}
\newcommand{\Var}{\mathrm{Var}}
\newcommand{\Cov}{\mathrm{Cov}}

\newcommand{\diff}{\mathrm{d}}

\newcommand{\bb}[1]{\left[#1\right]}
\newcommand{\bp}[1]{\left(#1\right)}
\newcommand{\bc}[1]{\left\{#1\right\}}
\newcommand{\bn}[1]{\left\|#1\right\|}
\newcommand{\ba}[1]{\left|#1\right|}

\RRHFirstLine{\theRUNAUTHOR}% default
\RRHSecondLine{\it\theRUNTITLE}% default
\LRHFirstLine{\theRUNAUTHOR}% default
\LRHSecondLine{\it\theRUNTITLE}

\begin{document}

\TITLE{Estimating Treatment Effects under Algorithmic Interference: A Structured Neural Networks Approach}
\RUNTITLE{Estimating Treatment Effects under Algorithmic Interference: A Structured Neural Networks Approach}

\ARTICLEAUTHORS{
\AUTHOR{\parbox[t]{.22\textwidth}{\centering Ruohan Zhan\\ \scriptsize{ruohan.zhan@ucl.ac.uk}\\
\scriptsize{University College London}}  
\parbox[t]{.2\textwidth}{\centering Shichao Han \\ \scriptsize{shichaohan@tencent.com}\\
\scriptsize{Tencent Inc.}}
\parbox[t]{.3\textwidth}{\centering Yuchen Hu \\ \scriptsize{yuchenhu@mit.edu}\\\scriptsize{Massachusetts Institute of Technology}}
\parbox[t]{.24\textwidth}{\centering Zhenling Jiang*\\
\scriptsize{zhenling@wharton.upenn.edu}
\\ \scriptsize{University of Pennsylvania}}
}
\RUNAUTHOR{}
}

\ABSTRACT{
Online user-generated content platforms allocate billions of dollars of promotional traffic through algorithms in two-sided marketplaces. 
To evaluate updates to these algorithms, platforms frequently rely on creator-side randomized experiments. 
However, because treated and control creators compete for exposure, such experiments suffer from algorithmic interference: exposure outcomes depend on competitors’ treatment status. We show that commonly used difference-in-means estimators can therefore be severely biased and may even recommend deploying inferior algorithms.
To address this challenge, we develop a structured semiparametric framework that explicitly models the competitive allocation mechanism underlying exposure. Our approach combines an algorithm choice model that characterizes how exposure is allocated across competing content with a viewer response model that captures engagement conditional on exposure. We construct a debiased estimator grounded in the double machine learning framework to recover the global treatment effect of platform-wide rollout. Methodologically, we extend DML asymptotic theory to accommodate correlated samples arising from overlapping consideration sets.
Using Monte Carlo simulations and a large-scale field experiment on a major short-video platform, we show that our estimator closely matches an interference-free benchmark obtained from a costly double-sided experimental design. In contrast, standard estimators exhibit substantial bias and, in some cases, even reverse the sign of the effect.

}

%%%%%%%%%%%%%%%%%% Title Page %%%%%%%%%%%%%%%%%%
\thispagestyle{empty}

\theARTICLETITLE

\theARTICLEAUTHORS

\vspace{20pt}

\begin{small}  
\noindent\theABSTRACT

\vspace{5pt}

\end{small}

\theARTICLERULE

\vspace{-60pt}

\blfootnote{We are grateful to Susan Athey, Tat Chan, Wanning Chen, Ryan Dew, Yue Fang, Bryan Graham, David Holtz, Ganesh Iyer, Ramesh Johari, Jussi Keppo, Hannah Li, Ilan Lobel, Xiaojie Mao, Tu Ni, Nian Si, K Sudhir, Yuyan Wang, Zikun Ye, Dennis Zhang, Jinglong Zhao,  conference participants at CODE@MIT 2023, Conference on AI/ML and Business Analytics 2023,  ACIC 2024, QME 2024, China India Insights Conference 2025, UTD FORMS Conference 2026, and seminar attendees at London Business School, University of California Berkeley, The Chinese University of Hong Kong, Carnegie Mellon University, New York University, University of Washington, National University of Singapore, Tsinghua University, The Chinese University of Hong Kong (Shenzhen) for their helpful comments. The short version of this paper has appeared in the ACM Conference on Economics and Computation (EC'24), and we are especially thankful to the anonymous reviewers for their constructive feedback. We are grateful to the WeChat Experimentation Team at Tencent for their generous support, particularly  Darwin Yong Wang, who has helped shape and advance this work's applications in the online platform.
}

\blfootnote{*Corresponding author. 
}

\newpage
\setcounter{page}{1}

\section{Introduction}

Online user-generated content platforms operate two-sided marketplaces connecting creators and viewers. One important source of monetization is through paid promotional traffic, which generates billions of dollars in platform revenue annually.  Creators pay for exposure, and the platform’s allocation algorithms determine which promoted content is shown to which viewers. Because these algorithms govern exposure, they directly shape both creator outcomes and platform revenue. As platforms routinely update these algorithms, evaluating proposed changes before deployment becomes a high-stakes and recurring managerial decision.

To evaluate such updates, platforms commonly rely on creator-side randomized experiments, in which creators and their content are randomly assigned to treatment or control algorithms, and downstream engagement outcomes are compared across groups. This design is natural in promotional settings because allocation policies, such as scoring adjustments, pacing rules and exposure caps, operate at the creator level and directly affect creator outcomes. 
However, the managerial question of interest is not the effect of changing the algorithm for an individual creator while holding others fixed. Rather, platforms seek to evaluate the global treatment effect (GTE) when the new algorithm is deployed to all creators. 
% Creator-side experiments, by contrast, operate in a mixed environment in which treated and control creators compete for exposure within the same consideration sets.

% We study this problem in the context of a large short-video platform where viewers are exposed to one video at a time, and promotional content in inserted when a viewer hits a promotional slot. We focus on the evaluation of a new promotional scoring algorithm. 
% The managerial question is whether the new algorithm should replace the status quo scoring rule for all creators. The platform addresses this question using a creator-side randomized experiment that compares performance under the new and existing algorithms during the experimental period.

% The creator ecosystem is highly fragmented, with many small advertisers competing for limited exposure and subject to budget and pacing constraints. 

We show that the standard difference-in-means (DIM) estimator, which compares outcomes between treated and control creators, can be severely biased due to algorithmic interference. 
When treated and control items compete for exposure, a change in the scoring rule for some creators mechanically affects the exposure received by others. This competitive allocation structure breaks the standard intuition that randomization guarantees unbiased comparisons. Consequently, an individual creator’s observed outcome depends not only on its own treatment assignment but also on the treatment status of others,  which violates the Stable Unit Treatment Value Assumption (SUTVA) \citep{imbens2004nonparametric}, and leads to interference bias \citep{bajari2021multiple,johari2022experimental,bright2022reducing,goli2023bias}. This bias can lead to wrong business decisions. In a large-scale field experiment we conducted on Weixin Channels, a leading short-video platform, we show that a standard DIM estimator concludes a new algorithm has a significant positive effect on an outcome, while the interference-free ground truth reveals a statistically significant negative effect. Relying on the flawed method would have caused the platform to deploy a worse algorithm.

The algorithmic interference generates two sources of bias in creator-side experiments. First, changes in the scoring rule shift exposure between treated and control creators within the experiment. If treated creators receive higher scores, they ``crowd out'' control items, leading to \emph{content exposure bias}.  
Second, promotional ranking systems are highly personalized. A scoring change is likely to have heterogeneous effects across viewers, leading to systematic differences in audiences exposed to treated vs. control content. Observed outcome differences therefore partially reflect shifts in viewer composition. We refer to this channel as \emph{viewer selection bias}. Importantly, both distortions arise from algorithm competition within the experimental environment and do not reflect the effect of global rollout.

The goal of this paper is to develop a reliable approach for estimating treatment effects in creator-side experiments that accounts for algorithmic interference. To evaluate whether a new promotional ranking algorithm should be adopted, we must compare two counterfactual scenarios: global treatment, in which all creators are scored under the new algorithm, and global control, in which all creators are scored under the status quo rule.
At a high level, our approach explicitly models the competitive allocation mechanism that gives rise to interference bias. The key component is a ``algorithm choice model'' that captures how exposure is allocated across competing content items, conditional on their treatment status, within each consideration set. 
We complement this with a viewer response model that predicts outcomes conditional on exposure. Treatment effects are then recovered by using the estimated models to simulate counterfactual exposure and response patterns under global treatment and global control.

The algorithm choice model is semi-parametric: its structural component enables counterfactual evaluation, while a flexible neural network learns the complex mapping from viewer-content pairs to scores. 
% By conditioning exposure probabilities on viewer characteristics, this model captures the systematic differences in audiences reached by treated versus control units. 
For the second component, a ``viewer response model'' uses another flexible neural network to predict the outcome for any given item-viewer pair once exposure occurs. Building on these components, we develop a ``debiased estimator'' of the global treatment effect. Because the allocation and response models are estimated using flexible machine learning methods, naive plug-in estimates can lead to biased inference. The debiased estimator corrects for this bias and delivers valid statistical inference even in the presence of complex, high-dimensional nuisance estimation  \citep{chernozhukov2018double, farrell2020deep}. 

A key methodological contribution of our paper is extending the DML framework beyond its canonical i.i.d.~samples to accommodate correlated data induced by overlapping consideration sets. When items appear repeatedly across viewer queries, their treatment status is shared rather than independently assigned, leading to sample dependence. We establish valid asymptotic inference for our debiased estimator under sample correlation. 
Our results rely on a mild exposure condition requiring that each item receive only limited average exposure across the viewer population. This assumption is natural in promotional traffic settings, where exposure is typically spread across many competing items.
This extension not only provides theoretical justification for our approach but also broadens the applicability of DML methods to other settings with correlated data, such as panel and marketplace environments.
% To do so, our proof models the data generating process sequentially and apply martingale limiting theorems \citep{hall2014martingale} to establish the estimator's asymptotics properties of an oracle estimator with known nuisances. Because of the Neyman orthogonality condition, the difference between the debiased and oracle estimators is shown to be small. This extension not only validates our approach but also broadens the applicability of the DML framework to other important settings with sample correlation, such as panel data.

We first evaluate our approach using Monte Carlo simulations and compare it against several benchmark methods. The difference-in-means (DIM) estimators directly compare outcomes between treated and control units. The pure deep learning estimator treats treatment effect estimation as a flexible prediction problem without imposing structure. The propensity-based methods, including IPW and AIPW, rely on reweighting based on the probability of observing the realized treatment vector within each consideration set. The simulation results confirm our theoretical predictions. The proposed debiased estimator successfully recovers the true treatment effect and provides valid statistical inference. In contrast, the DIM estimators and pure deep learning estimator are biased and fail to recover the true treatment effect. While the propensity-based estimators are unbiased in principle, their variance grows exponentially with the size of the consideration set, leading to unstable and unreliable estimates even when the set is only moderately large.

Beyond Monte Carlo studies, we conduct a large-scale experiment on the Weixin short-video platform to validate the performance of the proposed method. Besides running a creator-side experiment, our validation relies on a double-sided experiment \citep{ye2023cold,su2024long}. Specifically, we partition viewers and creators into three equal-sized distinct sub-universes, where viewers can only access creators within their own sub-universe. One sub-universe runs the creator-side experiment, while the other two operate fully under the treatment and control algorithms, respectively. While this double-sided design provides an interference-free ground truth because creators in the treated and control sub-universes do not compete for viewer attention, it is very costly for routine use: partitioning the market reduces statistical power, degrades market thickness, and requires substantial engineering resources. %Partitioning the platform for hundreds of concurrent experiments destroys statistical power. The resulting ``market thinning'' degrades the experience for both viewers and creators. Furthermore, implementing and maintaining these parallel sub-universes requires significant and ongoing engineering resources.

Using data from the standard creator-side experiment, we document clear evidence of both bias channels. Consistent with content exposure bias, treated items account for 56\% of realized exposures despite a 50\% assignment probability. In addition, we find systematic differences in the characteristics of viewers exposed to treated vs. control content, providing direct evidence of viewer selection bias. 
We apply our proposed estimator as well as the benchmark methods to data from the creator-side experiment, and compare the estimates to the ground-truth estimates from the costly double-sided design. Across three outcome measures, we find that the proposed estimator consistently yields results comparable to the ground-truth estimates. In contrast, benchmark estimators exhibit significant bias and, in some cases, even produce effects with the wrong sign, leading to incorrect business decisions. 

Our paper makes two primary contributions. Substantively, we provide a reliable method for evaluating promotional algorithms using standard creator-side experiments, allowing platforms to avoid costly deployment mistakes without resorting to expensive double-sided experimental designs. Methodologically, we introduce a structured neural network framework with debiased machine learning to deliver valid causal inference under algorithmic interference and correlated data. Together, these contributions provide both practical guidance for platform experimentation and theoretical advances for causal inference in algorithmic allocation environments.

\subsection{Related Literature}
\label{sec:related_work}

There is a growing literature studying interference in randomized experiments, where an individual's outcome is affected by others' treatments, violating the Stable Unit Treatment Value Assumption (SUTVA). Interference can lead to biased treatment effect estimates under standard A/B testing framework \citep{savje2021average,hu2021average,farias2022markovian,johari2022experimental,dhaouadi2023price,zhu2024seller}. To correctly estimate treatment effect, most of the existing literature leverages innovative \emph{experimental designs} to address the bias from interference. For example, one can use a clustered randomization design if interference primarily occurs within a cluster (e.g., on social networks) and then randomizes treatment at the cluster level \citep{holtz2020limiting,holtz2023reducing,ugander2013graph,hudgens2008toward}. Without a clear cluster structure, a switchback design can be applied that assigns a treatment condition randomly at the market level across different time periods \citep{bojinov2023design,hu2022switchback,ni2023design,xiong2023bias}. These experimental designs can be very useful but do not apply to all settings, including ours, where the content markets are too interconnected to form clusters and there exist significant time varying factors influencing the outcomes. 

Our work aligns more closely with an alternative approach that proposes innovative \emph{estimators} while leveraging standard experimental design. This requires domain-specific analysis of interference types. For example, in experiments for ranking algorithms, \cite{goli2023bias} address interference by leveraging items with positions close to those under the counterfactual ranking using historical A/B test data. In marketplaces where the platform matches supply and demand via linear programming, \cite{bright2022reducing} propose a shadow-price based estimator to mitigate interference bias. For interference bias that arises from consumers impacting others through limited capacity, \cite{farias2023correcting} model consumer behavior via a Markovian structure to address interference. Our paper proposes an estimator that applies to a common source of algorithmic interference where treated and control units compete for exposure in online platforms.

This paper draws from recent advances in double machine learning and semiparametric inference \citep{newey1994asymptotic,chernozhukov2018double,farrell2020deep,chernozhukov2019semi}.
In particular, the proposed treatment effect estimator directly builds on the doubly robust estimators for semi-parametric models proposed by \cite{farrell2020deep}, where the parametric outcome model is enriched by non-parametric components. We show that the inference results of the debiased estimator in \cite{farrell2020deep} apply to correlated samples, thus extending results from \cite{chernozhukov2018double}, which deals with i.i.d. samples. While in our setting, sample correlation arises from the overlapped items in consideration sets that share treatment status, our inference results can apply to other settings with correlated samples, such as panel settings with repeated observations from the same customer. Since many empirical data can be correlated, this inference result can broaden the applicability of the doubly robust estimator, which has already seen wide adoption in marketing and business research (e.g., \citealt{mummalaneni2022content}, \citealt{ye2023deep}, \citealt{kim2023tv}, \citealt{cheng2023selecting}, \citealt{ellickson2024using}).

\section{Background on Promotional Algorithm Experiments}
\label{sec:background}

We start by providing background on promotional algorithm experiments, which aim to assess whether a new promotional allocation  algorithm should be adopted. We then discuss double-sided randomization as a conceptual “oracle” design and highlight its various limitations.

\subsection{Algorithmic Treatment Effect}
\label{sec:recsys}

Online content platforms rely on algorithmic allocation systems to deliver promotional content to users. 
In our short-video context, once a viewer hits a promotional slot, the platform selects one video from a set of eligible videos depending on viewer and content characteristics. Operationally, the algorithm typically consists of multiple stages, including retrieval, ranking, and re-ranking, that progressively narrow the candidate set. In our short-video context, the algorithm selects a single item for each viewer query.\footnote{Our framework can be extended to slate listing via multiple-item choice models, where the core idea of interference applies. We focus on the single-item case, consistent with our empirical context.}

% , and this stage is frequently the focus of algorithm updates. 

The central managerial question is whether a proposed algorithmic update should replace the status quo allocation rule. Let the treatment denote the new algorithm and the control denote the existing algorithm. The object of interest is not the effect of treatment on an individual creator holding others fixed, but rather the effect of implementing the new algorithm platform-wide. Therefore, the \emph{treatment effect} of interest is the difference in total expected outcomes (aggregated across all viewer queries and content items) between two counterfactual environments: \emph{global treatment} where all creators receive traffic allocated  by the new algorithm vs. \emph{global control} where all creators receive   traffic by  the status quo algorithm. This estimand is commonly referred to as the Global Treatment Effect (GTE) \citep{johari2022experimental} or the Total Average Treatment Effect (TATE) \citep{goli2023bias}. We provide a formal definition using the potential outcomes framework in Section~\ref{sec:interference_bias}.

\subsection{Double-Sided Randomization as an Oracle Design}
\label{sec:double_sided_experiment}

A natural approach to estimating this treatment effect is to design an experiment that directly mimics the global treatment and global control policies. This can be achieved by a \emph{blocked double-sided design} that partitions the market into two separate sub-universes by randomizing both creators and viewers. In this design,  one sub-universe  contains a subset of viewers and creators operated entirely under  the treatment algorithm, while the other  contains a separate subset of viewers and creators operated under the control algorithm. By ensuring content from one sub-universe is excluded from the other, this design eliminates cross-group interference by construction. Comparing outcomes between these sub-universes therefore yields a direct, unbiased estimate of the treatment effect, making it a conceptual ``oracle'' for identification.

Despite being an interference-free benchmark, the double-sided design is very costly and rarely used for routine experiments for three main reasons. 
First is the statistical cost: platforms run hundreds of concurrent orthogonal experiments, and a double-sided design for each would require splitting the market into ever-smaller sub-universes. This causes sample sizes to decline exponentially, destroying the statistical power of all experiments. Second is the economic cost: partitioning the market leads to ``market thinning'', which degrades the experience for both parties. Creators receive less exposure and viewers are exposed to potentially less relevant content. Third is the engineering cost: implementing and maintaining multiple parallel versions of the entire platform with separate viewers, creators, and algorithms requires significant and ongoing operational resources, introducing complexity and system latency. Aside from these practical limitations, it is also unclear whether double-sided experiments are perfectly suited to recover the true global treatment effect for the entire marketplace. Because equilibrium outcomes may depend on the overall market size, estimates derived from these artificially partitioned small sub-universes may not generalize to the full platform. For these reasons, platforms often rely on the more practical creator-side experiments.

\subsection{Creator-side Randomization}

To evaluate promotional algorithms, creator-side randomization, in which creators (and their content items\footnote{In practice, randomization occurs at the creator level, with content inheriting the same status. Throughout the paper, we use ``creator'' and ``content items'' interchangeably.}) are assigned to treatment or control, is natural for two reasons. First, the primary objective in promotional settings is to measure the effect on creator-level outcomes, such as total engagement or return on promotion. Second, promotional allocation algorithms need to operate under creator-level constraints, such as budgets, pacing rules, or exposure caps, which makes it necessary to randomize across creators.

Indeed, viewer-level randomization, which assigns viewers to the treatment or control algorithms, does not apply to evaluate promotional algorithm updates. Both the outcome of interest and constraints are defined at the creator level and cannot be clearly evaluated or enforced by randomizing viewers. Moreover, viewer-level randomization introduces its own interference. Under creator-level budget or pacing constraints, exposure delivered to treated viewers mechanically affects the exposure available to control viewers. In addition, engagement from treated viewers changes content performance signals that affect subsequent exposure for control viewers, which is documented as symbiosis bias in literature \citep{brennan2025reducing}. For these reasons, viewer-level randomization is not suitable to evaluate promotional algorithm updates.

\section{Interference Bias from Creator Experiments}
\label{sec:interference_bias}

In this section, we first formalize the target estimand, the global treatment effect. We then describe  creator-side experiments and the standard difference-in-means (DIM) estimators that compare treated and control outcomes. Next, we discuss how interference arises and highlight two sources of bias: content exposure bias and viewer selection bias. We illustrate how these biases distort DIM estimators through both numerical examples and theoretical analysis. 

\subsection{Target Estimand}

The target estimand is the global treatment effect, which is the difference in aggregate outcomes under full deployment of the new algorithm relative to the status quo. We formalize this objective using the potential outcomes framework \citep{imbens2015causal}. Let $\calC$ denote the set of all content items. For each content item $c$, let $w_c\in\{0,1\}$ denote its treatment status,\footnote{For notational simplicity, we restrict attention to a binary treatment; the framework extends naturally to multiple treatments.} and let $\mathbf{w}_{-c}$ denote the treatment assignments of all other items. The full treatment assignment vector is $\textbf{w}$. For each content $c$, the potential outcome $r_c(\textbf{w})$ is determined by a viewer $v$ that generates a response $y$ (e.g., like or comment). We set $r_c(\textbf{w}) = y$ if content $c$ is exposed to viewer $v$, and $0$ otherwise.\footnote{We set $0$ as the default outcome because in empirical applications outcomes typically measure engagement (e.g., watch time or interactions), which is naturally zero when the item is not exposed to viewers. Other default values could be adopted in different applications without affecting the analysis.}

The value of a policy $\pi$, which specifies a treatment assignment rule, given viewer population $\bbP_v$, is defined as the expected outcome across all content items \citep{johari2022experimental, goli2023bias}:
\begin{equation}
    \label{eq:policy_value}
    Q(\pi; \bbP_v):=\sum_{c\in\calC} \bbE_{v\sim\bbP_v, \bfw\sim \pi}\bb{  r_c(\textbf{w})}. 
    % = \sum_{c\in\calC} \bbE_{v\sim \bbP_v, \bfw\sim \pi}\bb{y(v,c;\bfw)}.
\end{equation}
Let $\pi_\one$ denote the global treatment policy where $(w_c, \bfw_{-c})=\one$ and $\pi_\zero$ denote the global control policy where $(w_c, \bfw_{-c})=\zero$. The platform’s adoption decision depends on whether the policy value under global treatment exceeds that under global control. We define the treatment effect as:
\begin{equation}
    \label{eq:ate_definition}
    \tau:= Q(\pi_\one;\bbP_v)-Q(\pi_\zero;\bbP_v).
    % =\sum_{c\in\calC} \bc{ \bbE_{v\sim \bbP_v}\bb{y(v,c;\bfw=\one)} - \bbE_{v\sim \bbP_v}\bb{y(v,c;\bfw=\zero)} }.
\end{equation}
This estimand $\tau$ is the global treatment effect of the new algorithm, as it represents the expected impact of a global rollout across all creators.

\subsection{Creator-Side Experiments and DIM Estimators}

We describe the data structure from creator-side experiments. Each observation $i$ corresponds to a viewer query in which one content item is exposed, represented by the tuple $(V_i,\vv{C}_i,\vv{W}_i,k_i^*,Y_i)$. $V_i$ denotes the viewer characteristics. The consideration set $\vv{C}_i=\{C_{i,1},\dots,C_{i,K}\}$ contains the set of $K$ content items that reach the stage before randomization and are therefore unaffected by treatment assignment.\footnote{The model easily accommodates consideration sets of varying size; we assume a fixed $K$ for notational convenience.} $\vv{W}_i=\{W_{i,1},\dots,W_{i,K}\}$ records treatment assignment, with each $W_{i,k}$ drawn i.i.d.~from a Bernoulli distribution with treatment probability $q$. The algorithm exposes item $k_i^*$, and $Y_i$ is the viewer's response. We have a sample of $n$ viewer queries. 

% We assume that the distribution of consideration sets remains stationary under global treatment and control.\footnote{In nonstationary environments, one could instead use importance sampling to adjust for shifts in the distribution of consideration sets.}

Difference-in-means (DIM) estimators are commonly used to estimate treatment effects by comparing outcomes between treated and control items. We focus on two DIM estimators commonly used in the literature \citep{horvitz1952generalization,hajek1971comment}. Given $n$ samples ${(V_i,\vv{C}_i,\vv{W}_i,k_i^*,Y_i)}_{i=1}^n$ from a creator-side experiment with treatment probability $q$:
\begin{align}
   \text{Horvitz-Thompson Estimator:} \quad & \hat{\tau}^{HT-DIM}_n := \frac{\sum_{i=1}^n W_{i,k_i^*}Y_i}{nq} - \frac{\sum_{i=1}^n(1-W_{i,k_i^*})Y_i}{n(1-q)},\label{eq:HT-DIM}\\
\text{Hájek Estimator:} \quad &
    \hat{\tau}^{HA-DIM}_n := \frac{\sum_{i=1}^nW_{i,k_i^*}Y_i}{\sum_{i=1}^n W_{i,k_i^*}} - \frac{\sum_{i=1}^n(1-W_{i,k_i^*})Y_i}{\sum_{i=1}^n(1-W_{i,k_i^*})},  \label{eq:HA-DIM}
\end{align}
where $Y_i$ is the outcome for the exposed item $k_i^*$, and $W_{i,k_i^*}$ denotes its treatment status. The two estimators differ only in their normalization. The Horvitz–Thompson estimator normalizes by the assignment probability $q$, while the Hájek estimator normalizes by realized exposures. Next, we show why DIM estimators are biased under creator-side experiments.

\subsection{Interference and Sources of Bias}
\label{sec:recommender_interference}

In creator-side experiments, interference arises because treated and control items compete for exposure within the same consideration set. As a consequence, an item's outcome depends not only on its own treatment status but also on the treatment status of other items in the set, violating the Stable Unit Treatment Value Assumption (SUTVA). More specifically, interference leads to two forms of bias. First, suppose the treatment algorithm boosts the scores of treated items. Then these higher-scoring treated items are now more likely to be exposed, ``crowding out'' the control items. This leads to \emph{content exposure bias}: the share of treated items actually exposed to viewers differs from the treatment assignment probability. Second, since allocation algorithm is highly personalized, it might learn that the treatment is particularly effective for certain types of viewers (e.g., highly engaged or more valuable viewers). This leads to \emph{viewer selection bias}: the compositions of viewers exposed to treated and control items differ systematically. As a result, the outcomes of treated and control items are not directly comparable because they are drawn from different viewer populations.

Such interference directly affects both DIM estimators. With treated and control items exposed to systematically different viewers, both the Horvitz-Thompson and Hájek estimators are subject to viewer selection bias. In addition, the Horvitz-Thompson estimator is affected by content exposure bias, since it normalizes by the assignment probability $q$ rather than realized exposures. As we will show in Section \ref{sec:toy_example}, this does not imply that the Hájek estimator is less biased: the magnitude of bias depends on whether the two sources of bias act in the same or opposite directions.

\subsubsection{Illustrating the Bias with A Toy Example.}
\label{sec:toy_example}

We use a simple example to illustrate the two sources of interference bias. Suppose each item is assigned to the treatment or control with equal probability, and treated items receive a positive score uplift. Because treated and control items compete for exposure, treated items crowd out controls and are exposed more than 50\% of the time, despite 50/50 treatment assignment. This gap between assignment and realized exposure generates content exposure bias. Not accounting for the actual exposure leads to an upward bias in the treatment effect.

Now consider a scenario where the treatment uplift varies across viewers. Viewers differ in their baseline tendency to engage with promotional content (e.g., highly engaged viewers are more likely to make a purchase). Consider viewers with high vs. low baseline tendencies. When uplift is larger for high-baseline viewers (e.g., the algorithm prioritizes conversions), treated items are more likely to be exposed to high-baseline viewers relative to control items. Not accounting for the viewer selection bias leads to an upward bias of the treatment effect. 
The opposite happens when treatment uplift is larger for low-baseline viewers. Treated items are more likely to reach viewers with lower baseline tendency, and viewer selection bias leads to a downward bias in the treatment effect. 

\begin{figure}[t]
    \centering
     \caption{Illustrative Example of Inference Bias in Creator-Side Randomization}
    \includegraphics[width=.8\linewidth]{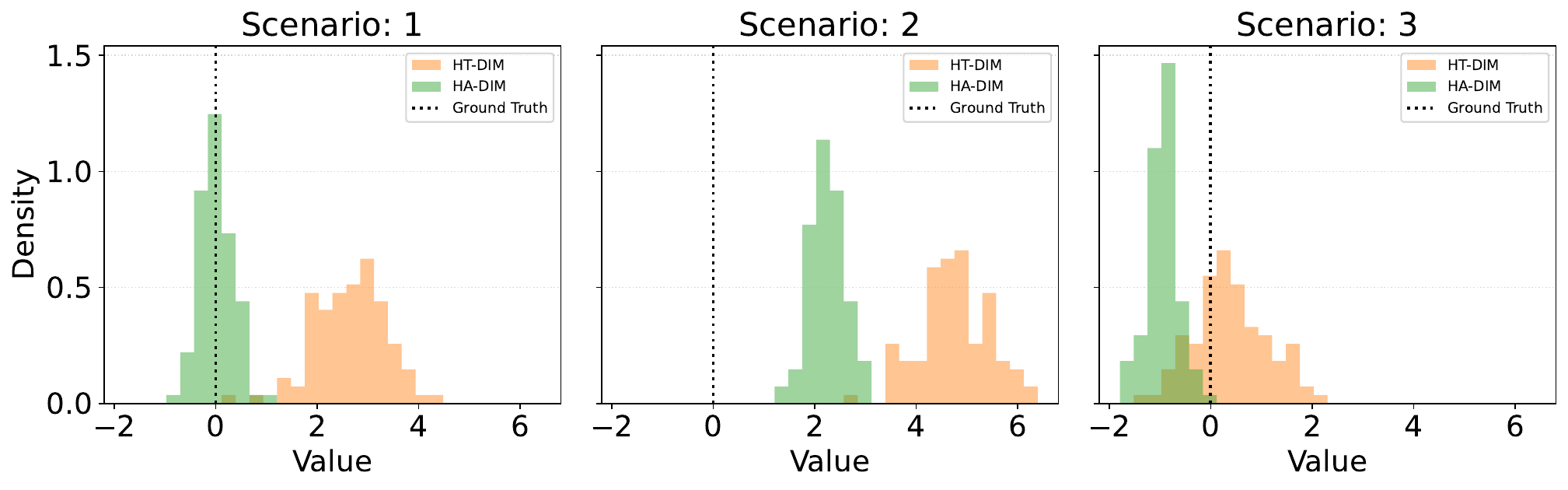}
    \label{fig:interference_bias}
    \vspace{0.5em}
    \begin{minipage}{0.9\linewidth}
    \footnotesize \textit{Notes:} HT-DIM = Horvitz–Thompson DIM estimator; HA-DIM = Hájek DIM estimator.\\
    Left figure: uplift independent of viewer type $\rightarrow$ positive exposure bias only. Middle figure: uplift larger for high-baseline viewers $\rightarrow$  positive exposure bias and positive viewer selection bias. Right figure: uplift larger for low-baseline viewers $\rightarrow$ positive exposure bias and negative viewer selection bias. 
    \end{minipage}
\end{figure}

Figure~\ref{fig:interference_bias} summarizes these patterns using a numerical example. In the left figure, when the uplift does not depend on viewer types, the Hájek estimator recovers the true effect by accounting for realized exposures, while the Horvitz–Thompson estimator overestimates due to content exposure bias. In the middle figure, when the uplift is larger for high-baseline viewers, both biases lead to upward-biased estimates. The Horvitz–Thompson estimator shows a larger bias since it reflects both item exposure and viewer selection effects, while Hájek estimator reflects only viewer selection bias. In the right figure, when the uplift is higher for low-baseline viewers, viewer selection bias and content exposure bias act in opposite directions. Hájek, affected only by viewer selection bias, underestimates the effect, while Horvitz–Thompson, affected by both, ends up slightly overestimating the effect.

%In Scenario II, where both biases favor the treated group, both estimators overestimate the effect, with Horvitz–Thompson showing a larger upward bias since it reflects both sources of bias, while Hájek reflects only viewer selection bias. In Scenario III, the two biases work in opposite directions: viewer selection bias favors the control group, leading Hájek to underestimate the effect, while content exposure bias favors the treated group, partly offsetting this negative bias. As a result, Horvitz–Thompson slightly overestimates the effect. The last row of Figure~\ref{fig:interference_bias} summarizes these patterns across the three scenarios.

%In Scenario I, where content exposure bias favors the treated group but no viewer selection bias is present, the Hájek estimator is unbiased while the Horvitz–Thompson estimator overestimates the treatment effect. 

% We consider three scenarios: uplift scales with item effects (Scenario I), with viewer effects (Scenario II), or inversely with viewer effects (Scenario III).  \emph{Viewer selection bias} arises only when uplift depends on viewer effects: in Scenario II, treated items are more likely to reach generous viewers, while in Scenario III they are more likely to reach picky viewers, leaving control items more exposed to generous viewers (second row of Figure~\ref{fig:interference_bias}).

\subsection{Theoretical Analysis of DIM Estimators}
\label{sec:bias_dim}
We now formalize the intuition developed above by characterizing the asymptotic limits of DIM estimators under creator-side experiments. We show how content exposure bias and viewer selection bias affect these estimators in large samples.

Let $e^*(q \mid v)$ be the probability that viewer $v$ is exposed to a treated item when the treatment assignment probability is $q$. The population-level probability of exposure is then $e^*(q) = \mathbb{E}_{v \sim \mathbb{P}_v} \left[ e^*(q \mid v) \right]$, where $\bbP_v$ is the distribution of viewers. Let $\bbP_{v}^{1,q}$  and $\bbP_{v}^{0,q}$ denote the distributions of viewers exposed to treated and control items, respectively. Content exposure bias occurs whenever $ e^*(q)\neq q$. Viewer selection bias arises whenever $\bbP_v^{1,q} \neq \bbP_v^{0,q}$, indicating that viewers exposed to treated vs. control differ systematically from the overall population.

We assume that no single item appears in an excessively large fraction of consideration sets (formal condition outlined in Assumption \ref{assump:item_appearance} in Appendix~\ref{appendix:dim_bias_proof}). This assumption is plausible for promotional content, which is subject to budget and traffic allocation constraints that naturally prevent any single item from dominating viewer queries.

% This condition implies that, on average, each item appears in no more than $o(n^{1/3})$ consideration sets. In practice, this assumption is typically satisfied because creator-side experiments often focus on promotional content, which is also our empirical setting, where make up only a relatively small portion of the total sample.

 \begin{theorem}[Bias of DIM Estimators]
 \label{thm:dim_bias}
 Suppose Assumption \ref{assump:item_appearance} holds. Under creator-side randomization with treatment probability $q\in(0,1)$, the Horvitz-Thompson estimator $\hat{\tau}^{HT-DIM}_n$ converges in probability to  $\tau^{HT}$, where
    \begin{align}
     \label{eq:ht_dim_convergence}
      \quad \tau^{HT} := \frac{e^*(q)}{q}  \sum_{c\in\calC}\bbE_{v\sim \bbP_{v}^{1,q}, \bfw\sim \calB(q)}  [r_c\mid w_{k^*}=1]  
 &- \frac{1-e^*(q)}{1-q}\sum_{c\in\calC}\bbE_{v\sim \bbP_{v}^{0,q}, \bfw\sim \calB(q)}[r_c\mid w_{k^*}=0],
    \end{align}
    with $w_{k^*}$ denoting the treatment status of exposed item and $\calB(q)$ denoting the Bernoulli trial with assignment probability $q$.
    The Hájek estimator $ \hat{\tau}^{HA-DIM}_n$ converges in probability to $ \tau^{HA}$, where
    \begin{equation}
      \label{eq:ha_dim_convergence}
      \tau^{HA} := \sum_{c\in\calC}\bbE_{v\sim \bbP_{v}^{1,q}, \bfw\sim \calB(q)}[r_c\mid w_{k^*}=1]  - \sum_{c\in\calC}\bbE_{v\sim \bbP_{v}^{0,q}, \bfw\sim \calB(q)}[r_c\mid w_{k^*}=0].
    \end{equation}
\end{theorem}

The proof is given in Appendix~\ref{appendix:dim_bias_proof}. For comparison, the true treatment effect is
\begin{equation}
\label{eq:ate_property}
  \tau =       \sum_{c\in\calC}\bbE_{v\sim \bbP_{v}, \bfw=\one}[r_c\mid w_{k^*}=1]  - \sum_{c\in\calC}\bbE_{v\sim \bbP_{v}, \bfw=\zero}[r_c\mid w_{k^*}=0].
\end{equation}
This result formalizes the two sources of bias discussed previously. Relative to the true treatment effect, the bias in the Hájek estimator is driven by the difference between the viewer populations ($\bbP_v^{1,q} \neq \bbP_v^{0,q}$), which is precisely viewer selection bias. The Horvitz-Thompson estimator is subject to this same bias, and is further affected by the deviation of the realized exposure from the assignment probability ($e^*(q) \neq q$), which is content exposure bias.

\section{Modeling Interference}
\label{sec:modeling}

In this section, we introduce a framework to model the interference pathway in creator-side experiments. An algorithm choice model captures how treated and control items compete for exposure, while a viewer response model describes outcomes once exposure occurs. Combining the two enables counterfactual analysis of alternative treatment policies to estimate treatment effects. We conclude by discussing the scope and robustness of the modeling approach.

\subsection{Algorithm Choice Model}

Upon receiving a viewer query, the algorithm chooses one content item for viewer $V_i$ from the consideration set $\vv{C_i}$. The evaluation of an item $C_{i,k}$ depends on its treatment status $W_{i,k}$, where $W_{i,k}=1$ indicates the treatment algorithm and $W_{i,k}=0$  the control algorithm. 

% Although it may seem feasible to simulate counterfactual recommendations offline by re-running the recommender system under alternative treatment assignments, such simulations generally fail to capture the realities of online serving. In practice, exposures are influenced by unobserved randomness, such as concurrent orthogonal experiments \citep{orthogonal,ye2023deep} and traffic-shaping mechanisms \citep{trafficshaping}, that cannot be reproduced offline. Recommendations depend on factors such as transient system conditions or user feedback that are typically not stored in system logs. As a result, offline recommendation simulations often diverge from actual online outcomes and are considered unreliable in industry practice \citep{bennett2007netflix}.

We approximate the algorithm allocation with a choice model built around a latent score $S_{i,k}$ of each item $C_{i,k}$ for viewer $V_i$. Intuitively, this score summarizes several stages of the algorithmic pipeline into a single measure that depends on viewer and content characteristics and the treatment status. The score functions are estimated to rationalize the observed exposure. Specifically, let the latent score $S_{i,k}$ be:
\begin{equation}
\label{eq:score}
    S_{i,k} = s_0(V_i, C_{i,k}) + W_{i,k}\cdot s_1(V_i, C_{i,k}) + \epsilon_{i,k}.
\end{equation}
The score function $s_0(\cdot,\cdot)$ represents baseline evaluation under the control algorithm, while $s_1(\cdot,\cdot)$ captures the treatment uplift. Both functions are parameterized by neural networks to flexibly approximate personalized exposures across diverse viewer–content pairs \citep{covington2016deep}. The error term $\epsilon_{i,k}$ accounts for the randomness in algorithmic allocation and the approximation errors of neural networks in representing the actual pipeline. We assume that the error term $\epsilon_{i,k}$ follows an i.i.d. Type 1 Extreme Value distribution, and the item with the highest score or ``utility'' in the consideration set is exposed to viewer $V_i$.

Under this formulation, the probability that item $k$ is exposed to viewer $V_i$ given consideration set $\vv{C_i}$ and treatment allocation $\vv{W_i}$ follows a multinomial logit form:
\begin{align}
\label{eq:choice_model}
% P_{i,k}\bp{V_i, \vv{C_i},\vv{W_{i}}; s_0, s_1}
p_{k}(V_i, \vv{C_i},\vv{W_{i}}; s_0, s_1):=\bbP\bp{k_i^*=k} 
= \frac{e^{s_0(V_i, C_{i,k}) +  W_{i,k}\cdot s_1(V_i, C_{i,k})}
}{\sum_{k'=1}^K e^{s_0(V_i, C_{i,k'}) +  W_{i,k'}\cdot s_1(V_i, C_{i,k'})}
}.
\end{align}
The choice model explicitly captures interference in the algorithmic allocation process. An item's exposure probability depends not only on its own treatment status but also on the treatment assignment of the other items in the set. 

The algorithm choice model is semi-parametric: it combines a structured logit choice model with flexible neural networks. The structural component enables counterfactual analysis under alternative treatment policies, which is difficult for fully nonparametric black-box models that often fail to generalize to counterfactual settings. 
The flexible neural network parameterization of utilities allows rich heterogeneity across viewer–item pairs. Thus, the model can capture complex viewer–content relationships, allowing the model to approximate the behavior of real-world algorithms. 
 
%Instead of proportional substitution across items, the neural networks allow a new item to disproportionately affect the exposure probabilities of items depending on content similarity and viewer-specific tastes.

Beyond enabling treatment effect estimation, the algorithm choice model can be practically valuable for platforms. By approximating the complex pipeline with a distilled model \citep{hinton2015distilling}, it can reduce latency when serving viewers online during peak traffic \citep{tang2018ranking}. Moreover, the learned score functions summarize the overall value of items to viewers, which can serve as a useful metric for evaluation.

\subsection{Viewer Response Model} 

To evaluate counterfactual treatment policies, we also need to model how viewers respond once an item is exposed. Outcomes of interest may include view time, likes, or conversions. We model the response $Y_i$ of viewer $V_i$ to exposed item $C_{i,k_i^*}$ as
\begin{equation}
\label{eq:viewer_outcome}
    Y_i = z(V_i, C_{i,k_i^*}) + \zeta_i,
\end{equation}
where $z(\cdot, \cdot)$ is a flexible function, parameterized by a neural net, and $\zeta_i$ is i.i.d.~noise. Separate response models can be trained for different outcomes.

This specification relies on two simplifying assumptions. First, conditional on the exposed item, outcomes do not depend on the item’s treatment status. This is reasonable since viewers are typically not aware of changes in algorithms. Second, it assumes that outcomes do not depend on prior content exposure.
In principle, the viewer response model could be extended to account for previous content interactions by incorporating past viewer consumption into the characteristics of the viewer query sample $V_i$, but doing so will introduce considerable complexity.

In practice, platforms often maintain pre-trained response predictors, which we leverage in our empirical setting. It is worth noting that our framework does not require these predictions to be unbiased. As detailed in Section~\ref{sec:main_estimator}, we introduce a debiased estimator to correct for potential errors in the nuisance predictions, which  is made possible by the treatment randomization in creator-side experiments.

\subsection{Counterfactual Analysis for Treatment Effect}
\label{sec:counterfactual}

The algorithm choice model in Equation~\eqref{eq:choice_model} and the viewer response model in Equation~\eqref{eq:viewer_outcome} together allow us to evaluate counterfactual treatment assignment policies. For any policy $\pi$, the expected outcome, or policy value, is (see Appendix \ref{appendix:policy_value} for detailed derivation):
\begin{equation}
\label{eq:policy_value_model}
   Q(\pi)=~ \bbE_{(V_i, \vv{C}_{i},\vv{W_i}\sim \pi)}\bb{ \sum_{k=1}^K z(V_i, C_{i,k}) \cdot p_{k}\bp{V_i, \vv{C_i},\vv{W_{i}}; s_0, s_1}}.  
\end{equation}
The inner term calculates the sum of each item's exposure probability multiplied by its expected response. The outer expectation is taken over the distribution of the viewer, consideration set, and treatment assignment.

The treatment effect can be estimated by comparing the policy values of global treatment $Q(\pi_{\one})$ and global control $Q(\pi_{\zero})$: 
\begin{equation}
\label{eq:ATE}
        \tau = Q(\pi_{\one})-Q(\pi_{\zero})= \bbE_{(V_i, \vv{C}_i)}\bb{\sum_{k=1}^K z(V_i, C_{i,k}) \cdot  \bp{p_{k}(V_i, \vv{C_i},\vv{W_{i}}=\mathbf{1}; s_0, s_1) - p_{k}(V_i, \vv{C_i},\vv{W_{i}}=\mathbf{0}; s_0, s_1) }}.
\end{equation}
Here the difference in exposure probabilities quantifies how the likelihood of showing each content item $C_{i,k}$ changes when going from global control to global treatment.

Finally, the framework is flexible to extend beyond the global treatment effect by comparing global treatment versus control. It can be applied to estimate heterogeneous effects on different types of creators. For a subgroup $\calC_0$, the corresponding subgroup treatment effect can be obtained by summing over items in the subgroup of interest for $C_{i,k}\in\calC_0$. 
% \begin{equation*}
% % \label{eq:cate_counterfactual}
%    \tau_{\calC_0}=\bbE_{(V_i, \vv{C}_i)}\bb{\sum_{k=1}^K 
%    I\bc{C_{i,k}\in\calC_0}\cdot z(V_i, C_{i,k}) \cdot \bp{p_{k}(V_i, \vv{C_i},\vv{W_{i}}=\mathbf{1}; s_0, s_1) - p_{k}(V_i, \vv{C_i},\vv{W_{i}}=\mathbf{0}; s_0, s_1) }},
% \end{equation*}
We can evaluate policies that allocate a fraction of traffic to the treatment algorithm (e.g., as the platform scales up the proportion of treatment algorithm). One can also compute the treatment effect in relative terms with $\frac{Q(\pi_{\one})-Q(\pi_{\zero})}{Q(\pi_{\zero})}$. 

\subsection{Scope and Assumptions of the Proposed Approach}
\label{subsec:scope}

Our framework is well suited for creator experiments occurring in the later, re-ranking stage of an algorithmic allocation system. At this stage, items compete for exposure within a relatively small consideration set, which allows for tractable choice model estimation. The approach is less suited for early retrieval stages that involve very large candidate pools where the large number of options makes the choice model difficult to estimate reliably. Furthermore, our analysis focuses on short-run outcomes, taking content and queries as given. We abstract away from other potential forms of interference, such as viewer-side temporal dynamics \citep{farias2023correcting} or long-run content supply effects. Capturing these important but distinct dynamics would require different modeling approaches and are left as avenues for future work.

The proposed choice model relies on two parametric assumptions. First, we assume the choice probabilities follow a multinomial logit form. This is a standard workhorse model in the discrete choice literature, arising from the assumption of Gumbel-distributed error terms on the each item's latent utility. While parametric, the resulting softmax function is highly flexible and can approximate any categorical distribution \citep{cervera2021uncertainty,ye2023deep}, making it a reasonable choice given the stochastic nature of online algorithmic allocation. Second, we assume that the score of each item  has an additive structure that depends only on its own characteristics and treatment status: $  \bar{S}_{i,k} = s_0(V_i, C_{i,k}) +W_{i,k}\cdot s_1(V_i, C_{i,k})$. This additive separability allows us to isolate the treatment uplift $s_1$ as a distinct function. A fully general model, where an item's score depends on the treatment status of all other items, would be much computationally costly, with the complexity scaling exponentially with the consideration set size. Given the strong empirical fit of the proposed specification (see Section~\ref{sec:empirical}), introducing additional flexibility is unlikely to materially improve performance while substantially increasing computational burden. 
% Crucially, this does not imply a homogeneous treatment effect. Because we parameterize both the baseline score $s_0$ and the uplift $s_1$ with flexible neural networks, our model can capture complex, non-linear treatment effect heterogeneity across different viewers and content items.

% From a practical standpoint, our framework is designed for the platform experimentation team. Its implementation requires access to the standard data logged during a creator-side experiment: the consideration set, the exposed item, the treatment assignment, and viewers and content features. A key premise of our approach is a simple offline simulation, even with the model's code, is fundamentally unreliable for answering the business question. An offline simulation fails to replicate a live algorithm's decisions because of factors missing from offline data, such as the inherent stochasticity from live exploration (e.g., testing new items) and concurrent experiments, and transient, real-time features of items, viewers, and system latency. Furthermore, evaluating algorithms requires outcome variables from real viewer interactions. Even if we are willing to use the imperfect simulated exposure and predicted responses, these proxies will likely introduce errors and biases, making it impossible to perform valid inference for treatment effect estimation. These limitations are precisely why platforms invest in costly live experiments.

A potential concern lies in possible errors in representing the algorithm choice model and the viewer response model. These models may be imperfect due to factors such as misspecification, omitted features, or neural network approximation errors. As we explain in Section~\ref{sec:main_estimator}, our framework is designed to be robust to such imperfections. The debiasing procedure exploits the randomness in treatment assignment to correct first-order bias arising from imperfectly estimated nuisance functions, thereby preserving the consistency and inference validity of the treatment effect estimate. However, this robustness depends on the nuisance models being ``good enough'' approximations. If they are severely misspecified, they may fail to achieve the convergence rate required for the debiasing theory to hold.

% depends only on its own characteristics and treatment status: $  \bar{S}_{i,k} = s_0(V_i, C_{i,k}) +W_{i,k}\cdot s_1(V_1, C_{i,k}).$ A fully general model would allow each item’s score to depend on the entire consideration set and treatment vector: $\bar{S}_{i,k} = s_k^*(V_i, \vv{C_i}, \vv{W_i})$. While more flexible, this quickly becomes computationally infeasible: with a consideration set of $K$ items, the additive model only requires $K$ baseline scores and $K$ uplifts, whereas the fully general model requires $2^K K$ scores, since each item’s score depends on $2^K$ possible treatment assignments.

% In practice, the additive specification is typically adequate. Empirically, it fits observed data well (see Section~\ref{sec:empirical_estimate}), and theoretically, recommendation scores primarily reflect viewer–content pairs rather than directly depend on other items’ treatment assignments. In applications where prior knowledge suggests richer dependencies, the score function can be adapted without requiring the full $2^K$ complexity. As long as treatment status enters parametrically, our estimation and inference framework in Section~\ref{sec:main_estimator} remains valid.

\section{Estimation and Inference Procedure}

In this section, we describe the estimation and inference procedure of the proposed model. We start with the estimation of the algorithm choice and viewer response models. We then construct a debiased estimator for the global treatment effect that remains valid even when the nuisances converge at slower-than-$\sqrt{n}$ rates. 
% Finally, we validate that the proposed approach can effectively recover the true point estimate and uncertainty via Monte Carlo simulations. 

\subsection{Model Estimation}

\label{sec:estimation}
We start by estimating the nuisance components that characterize the data-generating process as described in Section \ref{sec:modeling}. These nuisance components include the algorithm score components ($s_0$, $s_1$), which capture how the platform assigns exposure probabilities to content items under control and treatment, and viewer response $z$, which captures how viewers react once an item is exposed. We use deep feedforward neural networks with ReLU activation for their approximation guarantees \citep{farrell2020deep,farrell2021deep}. Below we describe the estimation procedure, the conditions required for identification, and the convergence guarantees.

The algorithm scores $(s_0,s_1)$ are estimated by minimizing cross-entropy loss, which is equivalent to maximizing the multinomial logit likelihood. Intuitively, this procedure learns the score functions so that the choice model best replicates actual exposures. For a viewer query $i$, we observe viewer~$V_i$ with consideration set $\vv{C_i}$ and treatment status $\vv{W_i}$ with realized exposure $k_i^{*}$. The loss is: 
\begin{align}
    \label{eq:cross_entropy}
    \ell_1(V_i,\vv{C_i},\vv{W_i}, k_i^*;&\tilde{s}_0, \tilde{s}_1):=~ -\log\bp{p_{k_i^*}\bp{V_i, \vv{C_i}, \vv{W_i};  \tilde{s}_0, \tilde{s}_1}}
    % =&~ -\bp{s_0(V_i, C_{i,k_i^*}) + W_{i,k_i^*} \cdot s_1(V_i,  C_{i,k_i^*}) } + \log\Big(\sum_{k=1}^K e^{s_0(V_i, C_{i,k}) + W_{i,k} \cdot s_1(V_i, C_{i,k})}\Big). \nonumber
\end{align} 
We estimate the nuisance components $\hat{s}_0, \hat{s}_1 \in\calF_{DNN}$ by minimizing the empirical loss:
\begin{align}
\label{eq:learn_score_functions}
    \bp{\hat{s}_0, \hat{s}_1} \in &\argmin_{\tilde{s}_0, \tilde{s}_1\in\calF_{DNN}}~\frac{1}{n}\sum_{i=1}^n \ell_1\bp{%V_i,\vv{C_i},\vv{W_i}, k_i^*;\tilde{s}_0, \tilde{s}_1
    \cdot}.
\end{align}

% Let $\calF_{DNN}$ denote the function class of such networks. We  obtain the nuisance estimates within $\calF_{DNN}$ by minimizing the specified loss functions. 

For the viewer response model $z$, we estimate the mapping from exposed viewer–item pairs to outcomes. When viewer $V_i$ gets exposed to content item $C_{i,k^*_i}$, we observe the outcome $Y_i$. Estimating the viewer response function is a standard machine learning task. For continuous outcomes, we use the mean squared error loss for any candidate function $\tilde{z}$:
\begin{equation}
\label{eq:loss2}
    \ell_2(V_i, C_{i,k^*_i}, Y_i; \tilde{z}) = \bp{\tilde{z}(V_i, C_{i,k^*_i}) - Y_i}^2.
\end{equation}
For categorical outcomes, one can use the cross entropy loss. With either type of loss function, we get the estimated $\hat{z}\in\calF_{DNN}$ by minimizing the loss: 
\begin{equation}
    \hat{z} \in \argmin_{\tilde{z}\in\calF_{DNN}} \frac{1}{n}\sum_{i=1}^n \ell_2(%V_i, C_{i,k^*_i}, Y_i; \tilde{z})
    \cdot).
\end{equation}

% \subsubsection*{Identification.}
To ensure that these nuisance functions are well defined, we impose an overlap condition (Assumption~\ref{assump:exposure_overlap} in Appendix~\ref{appendix:positive_exposure}), which ensures that each item in the consideration set has a positive probability of being exposed under both treatment and control. Under this assumption, the nuisance functions are identified, as formalized in Proposition~\ref{prop:identification} and proved in Appendix \ref{appendix:identification_proof}.
\begin{proposition}[Identification]
\label{prop:identification}
Suppose the algorithm choice behavior follows the semi-parametric form in Equation~\eqref{eq:choice_model}, and the viewer response model follows the nonparametric form in Equation~\eqref{eq:viewer_outcome}. Under Assumption \ref{assump:exposure_overlap}, the nuisance functions $(s_0,s_1,z)$ can be nonparametrically identified, up to a location normalization of the baseline score $s_0(V_i, C_{i,1})\equiv 0$.
\end{proposition}

% \subsubsection*{Convergence.}

\subsection{Debiased Estimator}
\label{sec:main_estimator}

Having estimated the nuisance components, we next construct an estimator of the treatment effect. We begin by describing a direct plug-in approach and then introduce our debiased (DB) estimator. Because the nuisance components estimated via neural networks generally converge at rates slower than $\sqrt{n}$,  the direct plug-in approach  generally fails to deliver valid statistical inference.
This limitation is important in practice because platforms do not base algorithm deployment decisions on point estimates alone. Reliable inference is needed to determine whether an estimated gain or loss is statistically meaningful and whether it could plausibly be zero or negative. Without valid uncertainty quantification, decision-makers risk acting on statistical noise.
The debiased estimator, by contrast, corrects for the bias induced by nuisance estimation errors through Neyman orthogonality and enables statistically valid inference. Formal results are presented in Section~\ref{sec:asymptotic}.

We start with a straightforward plug-in estimator.  With the estimated algorithm choice and viewer response models, one can directly plug in the estimated components into the GTE formulation in Equation~\eqref{eq:ATE}. For each viewer-consideration set pair $(V_i, \vv{C_i})$, the direct plug-in estimate $\mu$ is:
\begin{align}
%     &\mu(V_i,\vv{C_i};\hat{s}_0, \hat{s}_1,\hat{z}) = \mu^{(T)}(V_i, \vv{C_i};\hat{s}_0,\hat{s}_1,\hat{z}) -\mu^{(C)}(V_i, \vv{C_i};\hat{s}_0,\hat{s}_1,\hat{z}),\nonumber \\
%  \mbox{where}\quad &\mu^{(T)}(V_i, \vv{C_i};\hat{s}_0,\hat{s}_1,\hat{z}) = \sum_{k=1}^K \hat{z}(V_i, C_{i,k})
%      p_{k}(V_i, \vv{C_i}, \vv{W_i}=\one;  \hat{s}_0, \hat{s}_1), \label{eq:plugin_outcome}\\
% &\mu^{(C)}(V_i, \vv{C_i};\hat{s}_0,\hat{s}_1,\hat{z}) = \sum_{k=1}^K \hat{z}(V_i, C_{i,k})    p_{k}(V_i, \vv{C_i}, \vv{W_i}=\zero;  \hat{s}_0, \hat{s}_1).\nonumber
    &\mu(V_i,\vv{C_i};\hat{s}_0, \hat{s}_1,\hat{z}) = \sum_{k=1}^K \hat{z}(V_i, C_{i,k})
     \big[p_{k}(V_i, \vv{C_i}, \vv{W_i}=\one;  \hat{s}_0, \hat{s}_1) - p_{k}(V_i, \vv{C_i}, \vv{W_i}=\zero;  \hat{s}_0, \hat{s}_1)\big]
\end{align}
Averaging $\mu(\cdot)$ across the sample produces a plug-in estimator of the global treatment effect. If the components $({s}_0, {s}_1,z)$ are fully parametric, their estimates can achieve  $\sqrt{n}$-consistency and we can directly rely on the plug-in estimator for treatment effect estimation. However, with these components being approximated by neural networks, the direct plug-in is generally not $\sqrt{n}$-consistent and often converges at slower rates, resulting in biased inference.

To overcome the limitation of the direct plug-in estimator, we propose a debiased estimator that is $\sqrt{n}-$consistent and asymptotically normal. The key idea is to adjust the plug-in estimator with a correction term that removes first-order bias from nuisance estimation. Following \cite{farrell2020deep}, for each observation, the debiased estimator $\psi$ is: 
\begin{align}
\label{eq:debiase_estimate}
\psi_i^{DB}=~\mu(V_i,\vv{C_i}; \hat{s}_0, \hat{s}_1,\hat{z}) - \underbrace{\nabla\mu(V_i, \vv{C_i}; \hat{s}_0, \hat{s}_1,\hat{z})^T H(V_i,\vv{C_i};\hat{s}_0, \hat{s}_1,\hat{z})^{-1} \nabla\ell(V_i, \vv{C_i}, \vv{W_i}, k^*_i, Y_i; \hat{s}_0, \hat{s}_1,\hat{z})}_{\text{debiasing term}}.
\end{align}
The first term is the plug-in estimate based on the fitted nuisance functions, and the second term removes the first-order bias introduced by imperfect nuisance estimation. $\nabla\mu$ is the gradient of the plug-in estimator $\mu$ with respect to the nuisance components. $\nabla\ell$ is the gradient of the total loss $\ell$, which combines $\ell_1$ from the algorithm choice model and $\ell_2$ from the viewer response model, with respect to the nuisance functions. $H$ is the expected Hessian of the loss $\ell$ with respect to the nuisances, with the expectation taken over treatment assignment $\vv{W_i}$ under the creator-side randomization. Explicit expressions and derivations are given in Appendix \ref{appendix:debiased_term} and Appendix \ref{appendix:gradients}. 

Intuitively, the gradient $\nabla \mu$ captures the sensitivity of the treatment effect estimate to small perturbations in nuisance components.
The gradient $\nabla\ell$ reflects how nuisance estimation errors affect the total loss, which is scaled by the curvature of the loss function $H^{-1}$.\footnote{The validity of the debiasing step relies on the randomization of treatment assignment in creator-side experiments, which guarantees invertibility of the Hessian $H$ under Assumption \ref{assump:exposure_overlap} (bounded scores). Formal results and proofs on Hessian invertibility are provided in Appendix \ref{appendix:invertible_H}.} Together, these terms form the correction $\nabla \mu^{\top} H^{-1} \nabla \ell$, which offsets the first-order bias in the plug-in estimator and ensures that the final estimate is locally robust to nuisance estimation errors. 

 % This debiasing adjustment plays a role analogous to the propensity-score adjustment in standard causal inference.
 
% the double machine learning literature  constructs Neyman orthogonalized estimators~\citep{chernozhukov2018double,farrell2020deep}. 
% This Neyman orthogonality property ensures
% that errors in estimating nuisances impact the treatment effect estimation only to a second-order degree, allowing the treatment effect estimation to  achieve the $\sqrt{n}$-consistency and asymptotically normality for inference, even using nuisance estimates that converge slower than $\sqrt{n}$.

We next describe how the components of the debiasing term can be computed in practice. Both $\nabla\mu$ and $\nabla\ell$ can be explicitly written down and calculated (see Appendix~\ref{appendix:gradients} for details). When an explicit form is not available, these derivatives can be obtained through numeric differentiation. The expected Hessian $H$ can be explicitly computed in our case, unlike the general setup of \cite{farrell2020deep}, because the platform’s treatment assignment mechanism is known. But the computational complexity of explicit calculation scales exponentially with the size of the consideration set. When the size of the consideration set is manageable, as in the simulation studies in Section \ref{sec:sim}, we compute the exact expectation. When the consideration set is large, as in our empirical study in Section \ref{sec:empirical}, we approximate the expected Hessian via Monte Carlo simulations: specifically, we sample the treatment assignment vector $\vv{W_i}$ from the creator-side randomization and take the empirical mean of the resulting Hessians to estimate the expectation. 

The overall debiased (DB) estimator $\hat{\tau}_n^{DB}$ takes the average of $\psi_i^{DB}$ across all observations:
\begin{equation}
\label{eq:dr_estimator_general}
    \hat{\tau}_n^{DB} :=\frac{1}{n}\sum_{i=1}^n \psi_i^{DB}.
\end{equation}
The standard error of the DB etimator $\hat{\tau}_n^{DB}$ is estimated as $n^{-1/2}(\widehat{V}^{DB}_n)^{1/2}$, where the estimated variance $\widehat{V}^{DB}_n$ is computed as:
\begin{equation}
  \label{eq:dr_variance_general}
    \widehat{V}^{DB}_n =
    \frac{1}{n}\sum_{i=1}^n \bp{ \psi_i^{DB} - \hat{\tau}_n^{DB}}^2. 
\end{equation}

Following \cite{chernozhukov2018double} and \cite{farrell2020deep}, we employ sample splitting and cross-fitting to estimate the nuisance components. Specifically, the data are partitioned into $K$ folds. For each fold $k$, the nuisance functions $(\hat{s}_0^{(-k)}, \hat{s}_1^{(-k)}, \hat{z}^{(-k)}, \hat{H}^{(-k)})$ are estimated using all observations except those in fold $k$. The treatment effect is then estimated on the held-out fold to obtain $\hat{\tau}_n^{(k)}$. The cross-fitted treatment effect estimator is the average across $K$ folds: $\hat{\tau}^{DB}_n = \frac{1}{K}\sum_{k=1}^K \hat{\tau}_n^{(k)}$, with estimated variance  $\widehat{V}_n =\frac{1}{K}\sum_{k=1}^K \widehat{V}_n^{(k)}$.

\subsection{Asymptotic Results under Correlated Samples}
\label{sec:asymptotic}

We now turn to the asymptotic properties of the debiased estimator. We start with the convergence rates of the estimated nuisance components, then establish the Neyman orthogonality of the debiased estimator. We derive the $\sqrt{n}-$consistency and asymptotic normality, which together guarantee valid inference of the debiased estimator. A key contribution of our theoretical analysis is to extend double machine learning asymptotics to correlated data. 

We start with the convergence rates of the nuisance components estimated with neural networks. Leveraging recent results on the approximation power of deep neural networks \citep{farrell2020deep,farrell2021deep}, we provide convergence guarantees. The proof is provided in Appendix \ref{appendix:score_convergence}. 

\begin{proposition}[Convergence]
\label{prop:convergence}
Suppose the algorithm choice follows the semi-parametric form in Equation~\eqref{eq:choice_model}, and the viewer response model follows the nonparametric form in Equation~\eqref{eq:viewer_outcome} and is uniformly bounded. Suppose Assumption \ref{assump:exposure_overlap} and Assumption \ref{assump:smoothness} (in Appendix \ref{appendix:nuisance_estimation}) hold. Let $p$ be the smoothness parameter of the true nuisance functions $(s_0, s_1, z)$ and $d$ the dimension of the viewer-content covariates. If the neural networks $\hat{s}_0,\hat{s}_1, \hat{z}$ have width $J\asymp O(n^{d/2(p+d)}\log^2n)$ and depth $L \asymp \log n$, then for sufficiently large $n$,  with probability $1 - \exp(n^{d/(p+d)}\log^8 n)$,
\begin{equation*}
    \|\hat{s}_0 - s_0\|_{L_2} + \|\hat{s}_1 - s_1\|_{L_2}  + \|\hat{z} - z\|_{L_2} = O\bp{n^{-\frac{p}{p+d}}\log^8 n + \frac{\log\log n}{n}}.
\end{equation*}
\end{proposition}
Proposition~\ref{prop:convergence} characterizes the convergence rate of the nuisance components. For typical levels of smoothness $p$ and input dimension $d$, the convergence rate $n^{-\frac{p}{p+d}}$ lies between the parametric rate $n^{-\frac{1}{2}}$ and the slower nonparametric rate $n^{-\frac{1}{4}}$. This result is useful for two reasons. First, the direct plug-in estimator would inherit the slower convergence and therefore fail to achieve $\sqrt{n}$ consistency. Second, the debiased estimator is constructed to be orthogonal to first-order nuisance estimation error, so the established convergence rate ensures that it retains $\sqrt{n}$ consistency and valid asymptotic inference even when the nuisance components converge at a slower rate.

We next establish that the debiased estimator $\psi$ satisfies universal Neyman orthogonality \citep{chernozhukov2019semi,foster2023orthogonal}. This property ensures that small imperfections in the nuisance estimates have only a second-order impact on the debiased estimate. Neyman orthogonality is a key property for establishing the asymptotic normality of the debiased estimator. The proof is given in Appendix \ref{appendix:universal_orthogonality}. 

\begin{proposition}[Universal Orthogonality]
\label{prop:universal_orthogonality}
    The debiased estimator $\psi$, defined in Equation~\eqref{eq:debiase_estimate}, is universally orthogonal with respect to the nuisance components. Specifically, for any nuisance components $(\tilde{s}_0, \tilde{s}_1, \tilde{z}, \tilde{H})$, 
    \[
    \bbE[\nabla\psi(V,\vv{C}, \vv{W}, k^*, Y;\tilde{s}_0 = s_0, \tilde{s}_1 =s_1, \tilde{z}=z, \tilde{H}=H )\mid V, \vv{C}]=0,
    \]
    where $(V,\vv{C}, \vv{W}, k^*, Y)$ is drawn from creator-side experiments, and $\nabla\psi$ is the gradient with respect to the nuisances.
\end{proposition}

Leveraging the Neyman orthogonality property, we now show our main asymptotic results. The debiased estimator achieves $\sqrt{n}$-consistency and asymptotic normality even when the nuisance components converge at rates slower than $n^{-1/2}$. This central limit theorem thereby enables valid statistical inference and confidence interval construction for the treatment effect. The formal result is stated in Theorem~\ref{thm:clt} below.

 % \begin{theorem}[Asymptotic Normality of the Debiased Estimator]
\begin{restatable}[Asymptotic Normality of the Debiased Estimator]{theorem}{clt}
\label{thm:clt} Suppose Assumptions \ref{assump:item_appearance} \& \ref{assump:exposure_overlap}  hold. Assume that the data-generating process follows the algorithm choice model in Equation~\eqref{eq:choice_model} and the viewer response model in Equation~\eqref{eq:viewer_outcome}.
Suppose that the estimated nuisance functions are all bounded by the constant  $C$ in Assumption \ref{assump:exposure_overlap} and satisfy the convergence  rate:
$
    \|\hat{s}_0 - s_0\|_{L_2} +  \|\hat{s}_1 - s_1\|_{L_2} +\|\hat{z} - z\|_{L_2}
    = o(n^{-1/4})$.
    % and $\|\widehat{H} - H\|_{L_2} = o(n^{-1/4})$.
Then the debiased estimator $\hat{\tau}_n^{DB}$ defined in Equation~\eqref{eq:dr_estimator_general} is $\sqrt{n}$-consistent and asymptotically normal, with  estimated variance $\widehat{V}^{DB}_n$  in Equation~\eqref{eq:dr_variance_general}: 
\begin{equation*}
n^{1/2}\bp{\hat{\tau}_n^{DB} - \tau} /(\widehat{V}^{DB}_n)^{1/2} \Rightarrow \calN(0,1),    
\end{equation*}
implying that $\hat{\tau}_n^{DB} - \tau = O_p(n^{-1/2})$.
\end{restatable}
% \end{theorem}

Our main theoretical contribution lies in extending the asymptotic results of the debiased estimator to correlated samples. This extension is essential in our empirical application, where sample correlation arises from overlapping consideration sets. Even when each viewer query $i$ is treated as independent, samples are correlated because once an item appears in a consideration set, its treatment status in subsequent overlapping consideration sets becomes deterministic. Only items unique to a given consideration set follow independent Bernoulli randomization.  

To account for sample correlation, we model the data-generating process sequentially and employ martingale limit theorems to characterize the asymptotic behavior of the debiased estimator.
Assumption \ref{assump:item_appearance}, which requires the average number of appearances per item to grow at a rate slower than $O(n^{1/3})$,  controls the degree of sample correlation and ensures variance convergence under the martingale framework.
The proof proceeds by comparing the empirical debiased estimator $\hat{\tau}_n^{DB}$ with an oracle counterpart $\tilde{\tau}_n^{DB}$ that assumes the nuisance components are known without error. The oracle estimator forms a martingale difference sequence and hence converges to a normal distribution. By Neyman orthogonality, the discrepancy between the empirical and oracle estimators is asymptotically negligible, and the bound on their difference is shown to be sufficiently tight to establish the desired limit distribution. We provide detailed proof in Appendix \ref{appendix:clt}.

% Sample correlation can also impact the cross-fitting procedure. Overlapping consideration sets induce dependence between observations, violating the i.i.d.~assumption underlying standard cross-fitting procedure. To address this issue, one could apply the network cross-fitting of \citet{viviano2019policy}, that extends cross-fitting to dependent or networked samples. In practice, we adopt the standard cross-fitting procedure for estimating nuisance components, which is much easier to implement and does not materially affect inference performance in our empirical application.

\section{Monte Carlo Simulation}
\label{sec:sim}

In this section, we use Monte Carlo simulations to evaluate the performance of the proposed debiased estimator. We compare it against various benchmark estimators and show that, unlike the benchmarks, the debiased estimator delivers unbiased estimates and valid inference.

% We simulate data for $n=3000$ viewers. Each viewer has a two-dimensional feature $V_i \sim \mathcal{U}(0, 1)^2$ and a consideration set of $K=5$ items drawn from a pool of $500$ content items. Each item $C_k$ has a two-dimensional continuous feature $C_{1,k} \sim \mathcal{U}(0, 1)^2$ and a binary feature $C_{2,k} \sim \text{Bernoulli}(0.5)$. 
% Let the baseline score be $s_0(V_i,C_{k})=(V_i + C_{1,k}) + 0.05(V_i + C_{1,k})^2$, and the score uplift by the treatment algorithm be $s_1(V_i,C_{k})=\delta \cdot C_{2,k}\cdot (V_i + C_{1,k})$.  One item $k_i^\star$ is exposed according to the choice model in Equation~\eqref{eq:choice_model}. The viewer outcome model is $Y_i=(V_i + C_{1,k_i^*}) + \zeta_i$, where $\zeta_i \sim \mathcal{N}(0,0.1)$. 

% This setup generates the two forms of interference bias. Consider a case when  $\delta>0$. content exposure bias arises because treated items receive higher scores on average and therefore are more likely to be exposed than the treatment assignment probability. Moreover, viewer selection bias emerges because the treatment uplift $s_1(V_i,C_{i,k})$ depends on viewer features $V_i$. When $\delta>0$, viewers with higher baseline engagement (larger $V_i$) are systematically more likely to see treated items, leading to differences in the viewer populations exposed to treatment versus control.

\subsection{Benchmark Estimators}
\label{sec:benchmarks_sim}

We consider the following benchmark estimators. 

\textbf{Difference-in-means estimators.} 
As described in Section~\ref{sec:bias_dim}, we include two difference-in-means (DIM) estimators commonly used in practice: 
the \emph{Horvitz–Thompson DIM} estimator $\hat{\tau}^{HT-DIM}_n$, and 
the \emph{Hájek DIM} estimator $\hat{\tau}^{HA-DIM}_n$. 
Both estimators compare average outcomes between treated and control items but differ in their normalization.

\textbf{Pure deep learning based estimator.}
A natural alternative to the proposed semiparametric approach is to rely entirely on a flexible deep learning model and treat treatment effect estimation as a prediction task. Specifically, we train a neural network to directly learn the  outcome 
as a function of viewer characteristics, content characteristics, and the treatment assignment vector. The inputs coincide with those used in our proposed estimator; however, instead of imposing a  two-stage structured decomposition of the data generating process into a semi-parametric choice model and a viewer response model, we estimate the mapping nonparametrically and allow all inputs to enter a flexible neural network  without parametric structure. In other words, the network learns an unrestricted function mapping  ($V_i$, $\vv{C_i}$, $\vv{W_i}$) to the outcome  $Y_i$.
Once the model is trained, we simulate the global rollout counterfactual by replacing the observed treatment vector with the global treatment assignment ($\bfw_i=\mathbf{1}$) and the global control assignment ($\bfw_i=\mathbf{0}$). The treatment effect is then obtained by comparing the predicted outcomes under these two counterfactual environments.

We formalize this benchmark as a Pure Deep Learning (PDL) estimator. Let $\hat \mu(V_i,\vv C_i,\vv W_i)$ denote the trained neural network that outputs the predicted outcome. The PDL estimator of the treatment effect is defined as
\begin{equation*}
 \hat\tau^{PDL}=\frac{1}{n}\sum_{i=1}^n \left(\hat \mu(V_i,\vv C_i,\mathbf{1})-\hat \mu(V_i,\vv C_i,\mathbf{0})
 \right).   
\end{equation*}

\textbf{Propensity-based estimators.} 
We consider estimators based on importance weighting \citep{hahn1998role}. The key idea is to reweight each observation by the likelihood ratio (propensity) of the treatment assignment under a target environment (e.g., global treatment) relative to the experimental environment. 
In our setting, each outcome $Y_i$
depends on the entire treatment vector $\bfw_i = (W_{i1},\dots,W_{iK}) \sim \calB(q)$, where treatments are independently assigned to the $K$ items in the consideration set under the creator-side experiment.  Our goal is to compare  outcomes under two counterfactual policies: (i) global treatment ($\pi_\one$), under which all $K$ items are treated, and (ii) global control ($\pi_\zero$), under which all $K$ items are in control.  A natural approach is therefore to weight each observation by the likelihood ratio of the realized treatment vector $\bfw_i$ under the target policy relative to  the experimental design. This leads to the inverse-propensity-weighted (IPW) estimator
\begin{equation}
\label{eq:vanilla_ipw_estimator}
\hat{\tau}_n^{IPW}
= \frac{1}{n} \sum_{i=1}^n 
  \left[\frac{\bbP_{\bfw_i\sim\pi_\one}(\bfw_i=\one)}{\bbP_{\bfw_i\sim\calB(q)}(\bfw_i=\one) }
 - \frac{\bbP_{\bfw_i\sim\pi_\zero}(\bfw_i=\zero)}{\bbP_{\bfw_i\sim\calB(q)}(\bfw_i=\zero) }\right] Y_i
= \frac{1}{n} \sum_{i=1}^n 
  \left[\frac{I(\bfw_i=\one)}{q^K }
 - \frac{I(\bfw_i=\zero)}{(1-q)^K }\right] Y_i
\end{equation}

While this estimator removes the interference bias, its variance  is at the scale of $ \Theta\big(\frac{q^{-K} + (1-q)^{-K}}{n}\big)$ that increases exponentially with the consideration-set size $K$, making it unstable in practice even when $K$ is only moderately large. 
% We formalize this in Lemma~\ref{lemma:vanilla_ipw_estimator}, and provide the proof in Appendix \ref{appendix:vanilla_ipw_variance}.
% \begin{lemma}[Variance of IPW Estimator]
% Under creator-side randomization with treatment probability $q\in(0,1)$,
% the IPW estimator $\hat{\tau}_n^{IPW}$ is an unbiased estimator of $\tau$, and its variance satisfies \rzcomment{caveats to sample correlation which we may bound}
% \begin{equation}
% \Var(\hat{\tau}_n^{IPW}) = \Theta\left(\frac{q^{-K} + (1-q)^{-K}}{n}\right).
% \end{equation}
% \label{lemma:vanilla_ipw_estimator}
% \end{lemma}
To mitigate variance,  the augmented inverse-propensity-weighted (AIPW) estimator was proposed to incorporate regression adjustment \citep{hahn1998role}. The AIPW estimator takes the form:
\begin{equation}
\label{eq:vanilla_aipw_estimator}
\hat{\tau}_n^{AIPW}
= \frac{1}{n} \sum_{i=1}^n 
  \left[\hat{\mu}(V_i, \vv{C}_1, \one) + \frac{I(\bfw_i=\one)}{q^K }(Y_i - \hat{\mu}(V_i, \vv{C}_1, \one))
  \right] - \left[
  \hat{\mu}(V_i, \vv{C}_1, \zero) 
 + \frac{I(\bfw_i=\zero)}{(1-q)^K }(Y_i - \hat{\mu}(V_i, \vv{C}_1, \zero)\right],
\end{equation}
where $\hat{\mu}(V_i, \vv{C}_1, \one)$ and $\hat{\mu}(V_i, \vv{C}_1, \zero)$ denote predicted outcomes under global treatment and global control, respectively. In our implementation, these predictions are obtained from the pure deep learning model as described above.

Although AIPW typically improves finite-sample performance relative to IPW, the exponential variance problem persists. Because both IPW and AIPW rely on consideration-set-level propensities $q^K$
 and  $(1-q)^K$, their variance continues to scale exponentially in $K$. Thus, regression adjustment does not fundamentally resolve the instability induced by high-dimensional treatment vectors.

\subsection{Simulation Setup}

We design a simulation environment that mirrors the key features of creator-side experiments on content platforms.
For each of the $Q = 3000$ viewer queries, a consideration set of size $K$ is formed by sampling without replacement from $J = 500$ content items. We consider two scenarios with consideration set sizes $K=3$ and $K=8$. Each item $j$ has two attributes: a binary indicator $C_{j,1} \sim \mathrm{Bernoulli}(0.5)$ and a continuous attribute $C_{j,2} \sim \mathrm{Unif}(0,1)$. Each viewer $i$ has a preference parameter $V_i \sim \mathrm{Unif}(0,5)$, independently drawn across queries. Content items are randomly assigned to treatment or control with 50\% probability.

For viewer $i$ and item $k$ in the consideration set, the baseline score is  $s_0(V_i, C_k) = (V_i + C_{k,2}) + 0.1 (V_i + C_{k,2})^2$. Under treatment, the algorithm score adjustment is $s_1(V_i, C_k) = C_{k,1} \cdot (V_i + C_{k,2})$. Because this adjustment is heterogeneous across viewer--item pairs, it can lead to both content exposure and viewer selection biases. Exposure is determined by a multinomial logit model, and one item $k_i^*$ is exposed per query. The viewer response model is given by $Y_i = (V_i + C_{k_i^*,2}) + \varepsilon_i$, where $\varepsilon_i \sim \mathcal{N}(0,0.1)$. The viewer response depends on the viewer and the exposed item but not on treatment status.

Before evaluating the proposed and benchmark estimators, we compute the ground truth global treatment effect (GTE) using a large Monte Carlo sample of $100,000$ viewer queries. The GTE is defined as the difference in expected outcomes between global treatment, where all items are scored using $s_0+s_1$, and global control, where all items are scored using $s_0$. For $K=3$, the ground-truth GTE is approximately $-0.007$, and for $K=8$, the ground-truth GTE is approximately $0.033$. These values serve as benchmarks for evaluating estimator performance.

\subsection{Simulation Results}
\label{sec:simulation}

The Monte Carlo simulation runs $B=100$ independent replications. Each replication uses 3-fold cross-fitting to estimate nuisance components and treatment effects. For each estimator, we compute both the point estimate and its standard error.\footnote{The standard error for the DB estimator is calculated as $n^{-1/2}\widehat{V}^{1/2}_n$ following Theorem~\ref{thm:clt}, using 3-fold cross-fitting for nuisance estimation. For the Hájek DIM estimator, we use the plug-in variance from Theorem 1.2 of \cite{wager2024book}. For the remaining benchmarks, standard errors are obtained as the sample standard deviation of the estimates.} 
We compare the proposed debiased (DB) estimator $\hat{\tau}^{DB}_n$ with the benchmark estimators. Table~\ref{tab:montecarlo} reports the bias and uncertainty estimates for the debiased (DB) and benchmark estimators under consideration set sizes $K=3$ and $K=8$. Columns~(1) and (2) summarize estimator bias. Column~(1) reports the average deviation of the estimate from the GTE across simulation replications, and Column~(2) reports its standard error. 
%An unbiased estimator should have a mean bias that is statistically indistinguishable from zero. 
Column~(3) reports the Monte Carlo standard deviation of the estimator, and Column~(4) reports the average estimated standard error.

\begin{table}[htbp]
\centering
\caption{Point Estimate and Uncertainty Estimate}
\vspace{3mm}
\small
\begin{tabular}{lccccc}
\toprule
          & \multicolumn{2}{c}{Bias $(\hat{\tau}-\tau)$} & & \multicolumn{2}{c}{Uncertainty $\hat\sigma$} \\
\cmidrule(lr){2-4} \cmidrule(lr){5-6}
         & Mean & (Std. Err.) &  & Monte Carlo SD & Estimated SE\\
         & (1) & (2) & &  (3) & (4)\\
\cmidrule(lr){1-4} \cmidrule(lr){5-6}
% & \multicolumn{4}{c}{\textit{Small Consideration Set $K=3$}} \\
\addlinespace[0.25em]
Debiased (proposed) \\
\addlinespace[0.25em]
DB ($K=3$)      & -0.0005 & (0.0009) &  & 0.0091 & 0.0102 \\
DB ($K=8$)     &  -0.0004 & (0.0013) &  & 0.0132 & 0.0151 \\
\addlinespace[0.25em]
Difference-in-means \\
\addlinespace[0.25em]
HT-DIM ($K=3$)  & 2.2704 & (0.0353) &  & 0.3530 & 0.1161 \\
HT-DIM ($K=8$) & 4.3238 & (0.0253) & &  0.2529 & 0.0972 \\
\addlinespace[0.25em]
HA-DIM ($K=3$) & 0.2907 & (0.0058) &  & 0.0585 & 0.0380 \\
HA-DIM  ($K=8$) &  0.8907 & (0.0084) & &  0.0843 & 0.0365 \\
\addlinespace[0.25em]
Pure deep learning \\
\addlinespace[0.25em]
PDL ($K=3$)    & 0.0081 & (0.0017) &  & 0.0169  & 0.0008 \\
PDL ($K=8$)   & 0.0156 & (0.0031)  & & 0.0310  & 0.0011 \\
\addlinespace[0.25em]
Propensity-based \\
\addlinespace[0.25em]
IPW ($K=3$)    & 0.0772 & (0.1005)  & & 1.0052 & 0.2472 \\
IPW ($K=8$)   & -0.0987 & (0.2311) & & 2.3117 & 1.3774 \\
\addlinespace[0.25em]
AIPW ($K=3$)   & 0.0098 & (0.0112) &  & 0.1117 &  0.1510\\
AIPW ($K=8$)   & 0.0363 & (0.0622) & &  0.6218 & 0.8339 \\
% \addlinespace[0.25em]
% Propensity-based \\
% GIPW    & 0.0026 & (0.0006) & 0.0133 & 0.0078 \\
% GAIPW   & 0.0021 & (0.0004) & 0.0098 & 0.0007 \\
% \midrule
% \addlinespace[0.25em]
% & \multicolumn{4}{c}{\textit{Large Consideration Set $K=8$}} \\
% \addlinespace[0.25em]
% Debiased (proposed) \\
% \addlinespace[0.25em]
% Difference-in-means \\
% \addlinespace[0.25em]
% Pure deep learning \\
% \addlinespace[0.25em]
% Propensity-based \\
% GIPW    & -0.0009 & (0.0005) & 0.0119 & 0.0075 \\
% GAIPW   & -0.0008 & (0.0004) & 0.0090 & 0.0007 \\
\bottomrule
\label{tab:montecarlo}
\end{tabular}
\end{table}

The DB estimator performs well across all dimensions. Its average bias is effectively zero and small relative to sampling variability. Moreover, its estimated standard errors closely match the true standard errors approximated by the Monte Carlo variability across replicates, supporting valid statistical inference. These results confirm that the debiasing step successfully eliminates first-order bias from nuisance estimation and properly accounts for dependence induced by overlapping consideration sets.

In contrast, both DIM estimators show substantial bias due to interference under creator-side randomization. In addition, their reported standard errors underestimate true variability because they ignore the correlation induced by shared treatment assignments across overlapping consideration sets. As a result, they fail both in point estimation and inference. The PDL estimator also exhibits systematic bias. Because the model is trained flexibly in experimental environment with both treated and control items, it does not correctly extrapolate to the counterfactual global treatment and control policies. Its estimated standard error estimates are also unreliable.

Propensity-based estimators (IPW and AIPW) are unbiased in expectation. However, since their inverse propensity weights scale with $q^{-K}$ and $(1-q)^{-K}$, variance increases exponentially in $K$. Even at $K=3$, their Monte Carlo standard deviation is large. Although AIPW partially mitigates bias through regression adjustment, it does not resolve the fundamental variance explosion. The large variance becomes worse as sample size increases from 3 to 8, which makes it impractical to use with even a moderate-sized consideration set. Their estimated standard errors also deviate from true variability in finite samples.

Taken together, the DB estimator has negligible bias and accurate uncertainty quantification. All other benchmark estimators fail to recover the true GTE. It is also worth mentioning that most benchmark estimators underestimate uncertainty, largely because they ignore sample correlations arising from overlapping items in consideration sets that share the same treatment status. By contrast, the DB estimator properly accounts for this dependence structure and achieves the desired asymptotic properties through Neyman orthogonality (Section~\ref{sec:main_estimator}).

% \begin{figure}[ht]
% \centering
%     \caption{Monte Carlo Results: Point Estimate and Standard Error}
%     \vspace{3mm}
%     \includegraphics[width=0.49\linewidth]{graphs/ms_revision/estimate_new.pdf}
%     \hfill
%     \includegraphics[width=0.49\linewidth]{graphs/ms_revision/variance_new.pdf}
% \label{fig:syn_ab2}
% \end{figure}

\section{Empirical Application}
\label{sec:empirical}

In this section, we present results from a large-scale field experiment. The results show that the treatment effect estimates from the proposed debiased estimator closely align with the approximate ground truth, obtained from the costly double-sided design, whereas estimates from other methods deviate more substantially and in some cases even have the opposite sign.

\subsection{Institutional Background}

We conducted a large-scale field experiment in collaboration with \emph{Weixin Channels}, the short-video platform embedded in China’s largest mobile messaging app. When a viewer reaches a promotional ``slot'', the algorithm selects from the pool of promoted videos. Each viewer query thus results in the exposure of a single video, making it directly consistent with our algorithm choice model.

The treatment is a new promotional algorithm, while the control is the status quo. The intervention occurs at the re-ranking stage of the allocation pipeline where the consideration sets, consisting of videos that have reached the stage prior to the intervention, are relatively small (see Section~\ref{subsec:scope}). Interference arises when treated videos, using treatment algorithm for scoring, and control videos, using control algorithm, appear in the same consideration set competing for exposure.

We discuss the plausibility of the set of assumptions in our empirical context. The controlled appearances Assumption~\ref{assump:item_appearance} ensures that the same item does not repeatedly appear in many consideration sets on average. The bounded scores Assumption~\ref{assump:exposure_overlap} ensures that every video has a positive probability of being exposed. In our empirical setting, both assumptions are reasonable since promoted videos are subject to budgeting and delivery constraints, so that each item has positive exposure and no item can dominate a very large number of user queries.\footnote{In our data, the most popular item appears in only 2.33\% of consideration sets, while about 90\% of items appear in no more than $5.28 \times 10^{-5}$ of total viewer queries.} Furthermore, viewers rarely encounter two promoted videos in close succession, so temporal interference across queries is likely limited.

\subsection{Experimental Design} 
\label{sec:exp_design}

To evaluate the validity of treatment effect estimators, we simultaneously conduct a creator-side experiment and a blocked double-sided experiment, with the latter serving as an interference-free benchmark. We track three key outcomes of interest. While we cannot disclose their exact definitions for confidentiality reasons, they are all binary measures related to viewer engagement.

We randomly divide viewers and creators into three equal-sized sub-universes to mitigate potential market-size effects. Each sub-universe is blocked from the others, so viewers can only see videos from creators within the same sub-universe. The first sub-universe hosts the creator experiment, where half of the creators are randomly assigned to the treatment algorithm and half to the control. The second and third sub-universes implement the double-sided randomized experiment: one uses the treatment algorithm, while the other uses the control algorithm. Because the two sub-universes are blocked, treated and control videos never appear in the same consideration set, eliminating interference. As a result, the double-sided design provides an interference-free benchmark for the true treatment effect, denoted as $\hat{\tau}^{DS}$. While effective, the double-sided design is costly to implement and thus impractical for evaluating all creator-focused treatments at scale (see Section~\ref{sec:double_sided_experiment}).

The experiments from all three sub-universes were conducted simultaneously to avoid time-varying differences in creator outcomes. The experiment ran for eight days, from May 27 to June 3, 2024, encompassing approximately 1 million videos, evenly split across the three sub-universes. For each viewer query $i$, we observe $(V_i,\vv{C}_i, \vv{W}_i, k_i^*, Y_i)$, where $V_i$ and $C_{i,k}$ include both raw features and learned embeddings from historical data. The treatment assignment vector $\vv{W}_i = (W_{i,1}, \dots, W_{i,K})$ indicates whether the video was scored using the treatment ($W_{i,k}=1$) or control ($W_{i,k}=0$) algorithm.

\subsection{Randomization Check and Evidence of Interference}
\label{sec:evidence_interference}

For randomization check, we compare the distributions of several key covariates identified as strong predictors of the outcomes of interest in historical data. While we cannot disclose their exact definitions, these covariates are commonly examined by practitioners to ensure that randomization was implemented properly. The set includes video features, users’ historical engagement, pre-treatment values of the three outcome variables, and a system-level performance metric. Table~\ref{tab:balance} reports p-values from balance tests on six categorical covariates and four continuous covariates. The ``creator experiment” and ``double-sided experiment” rows test covariate balance between treatment and control groups within each experiment, while the ``across experiments” row compares covariates between the two experimental designs. None of the covariates show significant imbalance, confirming that randomization was implemented successfully.

\begin{table}[h]
\centering
\caption{P-values from Randomization Check}
\small
\label{tab:balance}
\begin{tabular}{lcccccc}
\toprule
 & \multicolumn{6}{c}{Categorical Covariates} \\
\cmidrule(lr){2-7}
 & Cat.1 & Cat.2 & Cat.3 & Cat.4 & Cat.5 & Cat.6 \\
\midrule
Creator experiment & 0.279 & 0.789 & 0.462 & 0.770 & 0.689 & 0.393 \\
Double-sided experiment & 0.137 & 0.102 & 0.976 & 0.200 & 0.165 & 0.067 \\
Across experiments      & 0.792 & 0.168 & 0.047 & 0.540 & 0.056 & 0.461 \\
\midrule
 & \multicolumn{6}{c}{Continuous Covariates} \\
\cmidrule(lr){2-7}
 & Cont.1 & Cont.2 & Cont.3 & Cont.4 &  &  \\
\midrule
Creator experiment & 0.580 & 0.403 & 0.521 & 0.549 &  &  \\
Double-sided experiment & 0.186 & 0.225 & 0.771 & 0.630 &  &  \\
Across experiments      & 0.537 & 0.921 & 0.461 & 0.819 &  &  \\
\bottomrule
\end{tabular}
\end{table}

Using data from the subworld implementing the standard creator-side experiment, we provide direct empirical evidence of interference operating through two channels: content exposure bias and viewer selection bias. First, although treatment assignment is balanced at 50\% treatment and 50\% control, realized exposure shares deviate from this benchmark. In actual exposures, 56\% of actual exposures go to treated items, with the remaining 44\% being control items. This discrepancy between assignment probability and exposure probability provides direct evidence of content exposure bias.
Second, we observe systematic differences in the characteristics of viewers exposed to treated versus control content. Across four key covariates that are all strong predictors of downstream outcomes, the relative differences between treated and control exposures, measured as (treat-control)/control, are -7.175\%, -8.017\%, -7.994\%, and +8.619\%. All these differences are statistically significant. This indicates that treated and control content reach systematically different viewer populations, providing direct evidence of viewer selection bias. 

\subsection{Estimation Details}
\label{sec:empirical_estimate}
We describe the implementation details of the proposed debiased estimator using data from the creator-side experiment. Estimating the debiased estimator in Equation~\eqref{eq:debiase_estimate} requires several components: the nuisance functions $\hat{s}_0$ and $\hat{s}_1$ from the algorithm choice model, the viewer response model $\hat{z}$, and the gradients and Hessian matrix to construct the debiasing term. We discuss each component in turn. 

We start with the algorithm choice model. The consideration set is defined by empirically identifying the smallest set of items that could realistically be exposed prior to treatment. Figure~\ref{fig:size_rank} shows the empirical exposure probability of items based on their ``pre-treatment'' ranking by the algorithm. Exposure probability is highest for the top-ranked item and quickly declines to nearly zero. Based on this pattern, we set the consideration set size to $K=15$, which captures nearly all items with non-negligible exposure probability. This choice ensures that the consideration sets cover relevant items while not being too large. 

\begin{figure}
    \centering
    \caption{
    Exposure Probability of Items with Pre-treatment Ranking
    }
    \includegraphics[width=.6\textwidth]{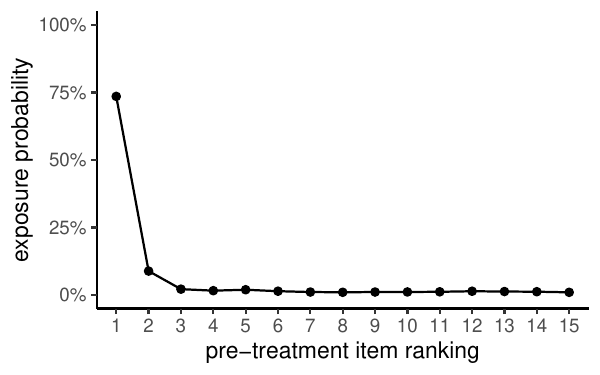} 
    \label{fig:size_rank}
\end{figure}

The algorithm choice model is estimated using data $(V_i, \vv{C}_i, \vv{W}_i, k^{\star})$. Figure~\ref{fig:choice_model_nn} illustrates the architecture of the semiparametric model. The viewer and content embeddings $(V_i, C_{i,k})$ pass through a black-box component that outputs two quantities: the baseline score $s_0$ and the treatment uplift $s_1$. These scores are then passed to a model layer that computes the exposure probabilities for all items in the consideration set. 

To compute the scores, we leverage the platform’s pre-trained representation model $\phi(\cdot)$ to generate features for each viewer–content pair $X_{i,k} := \phi(C_{i,k}, V_i)$. The representation model $\phi$, which is trained on historical data, captures a range of expected evaluations such as predicted viewer engagement. These features are then passed through two fully connected layers that output the baseline score and the treatment uplift. During estimation, the pre-trained representation model is held fixed, while the fully connected layers are updated. This approach can be viewed as fine-tuning the pre-trained model for our specific estimation task. Leveraging these pre-trained features substantially reduces both data and computational demands relative to training from raw embeddings.

In the model layer, the treatment assignment vector $\vv{W}_i$ enters explicitly, and exposure probabilities for all items in the consideration set are computed following Equation~\eqref{eq:choice_model}. The nuisance components $\hat{s}_0$ and $\hat{s}_1$ are estimated by maximizing the likelihood that item $k^{\star}$ is selected for viewer $V_i$ from the consideration set $\vv{C}_i$, conditional on treatment assignment $\vv{W}_i$. The model is trained with categorical cross-entropy loss as specified in Equation~\eqref{eq:cross_entropy}.

\begin{figure}
     \centering
        \caption{Semiparametric Algorithm Choice Model}
        \includegraphics[scale=.55]{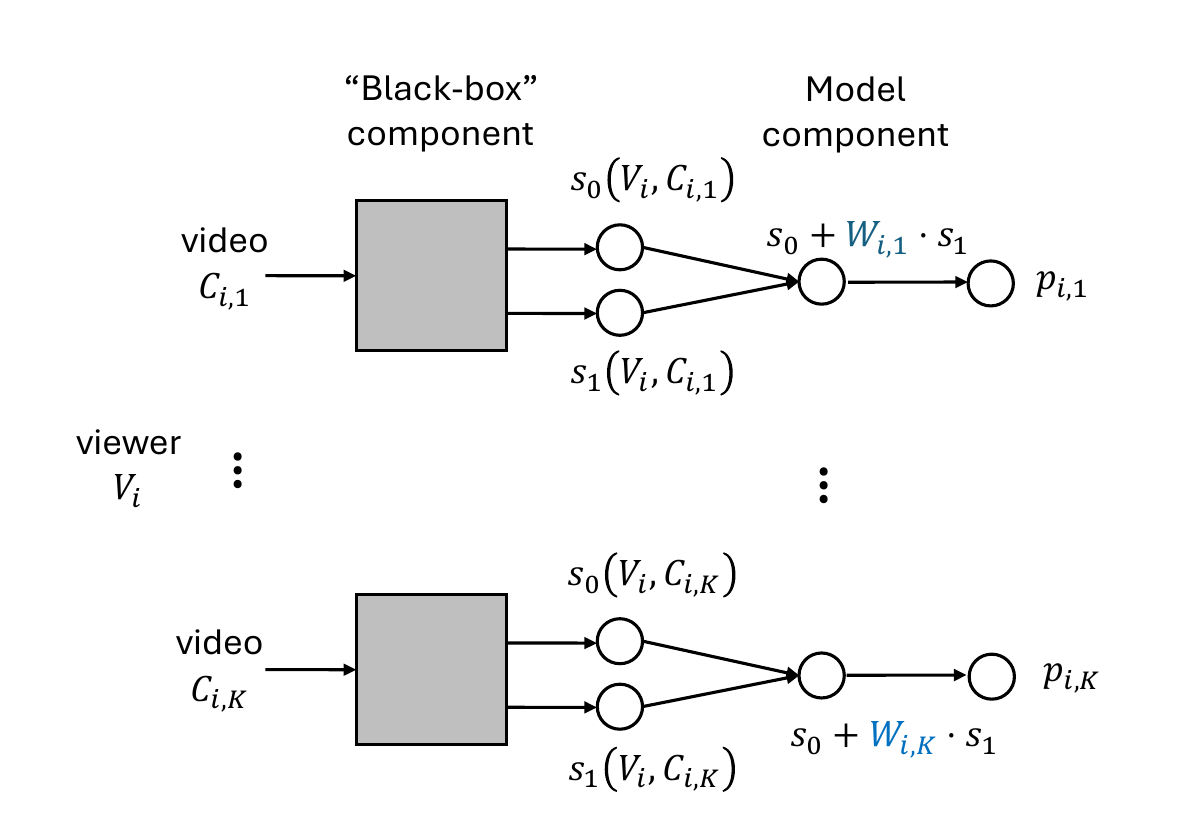}
        \label{fig:choice_model_nn}      
\end{figure}

We evaluate the fit of the algorithm choice model by comparing its predictions with the actual exposed items. Figure~\ref{fig:roc_empirical} presents the ROC curve. The dashed diagonal line corresponds to a non-informative classifier with an AUC of 0.5, whereas our fitted model achieves an AUC of 0.97. The very high out of sample predictive accuracy indicates that the structured logit–neural network specification can approximate the platform’s actual allocation behavior well, and that the imposed additive and logit structure does not materially restrict flexibility in practice.
% It is worth noting that our network structure, which builds upon a pre-trained representation model, may not strictly satisfy the width and depth conditions specified in Proposition~\ref{prop:convergence}. These conditions guarantee that the estimated nuisance components converge faster than $n^{-1/4}$, a key requirement for the asymptotic normality and inference guarantees established in Theorem~\ref{thm:clt}. Nonetheless, the model's strong empirical performance suggests that the convergence rate is likely sufficiently fast in practice.  %This strong performance is achieved by leveraging pre-trained models without requiring an overly complex neural network.

\begin{figure}
 \centering
 \caption{ROC Curve of the Fitted Algorithm Choice Model}
    \includegraphics[scale=0.6]{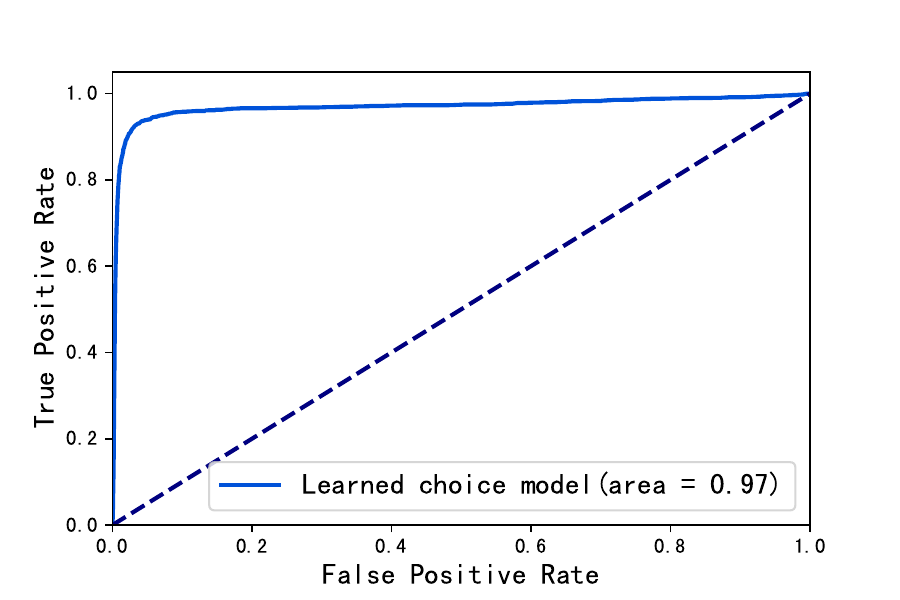} 
    \label{fig:roc_empirical}
\end{figure}

%Moreover, in the treatment effect estimation results, the confidence intervals produced by the proposed DB estimator closely overlap with the ground truth, further indicating that the estimated nuisance model is adequate for the inference results in Theorem~\ref{thm:clt} to hold in practice.

For the viewer response model, we leverage the platform’s pre-trained models that predict various outcomes for each viewer-content pair. Since the experiment tracks three outcomes of interest, each uses a separate viewer response model, while they all share the same algorithm choice model. The observed viewer responses from the experiment are used to compute the loss specified in Equation~\eqref{eq:loss2}, which contributes to the construction of the debiasing term.

The debiasing term contains the gradients of the plug-in estimator and the loss function with respect to the nuisance components, as well as the Hessian of the loss function. The gradients can be computed directly in our setting, as detailed in Section~\ref{sec:main_estimator}. The expected Hessian matrix of the loss function is taken with respect to the treatment assignment. To approximate this expectation, we draw 500 realizations of the treatment assignment and compute the empirical mean of the corresponding Hessians across these samples.

\subsection{Empirical Results}
\label{sec:treatment_effect_estimation}

We evaluate the debiased estimator using data from the field experiment. Specifically, we estimate the treatment effect from the creator-side experiment using the debiased estimator alongside all benchmark estimators. 
Since the propensity-based methods (IPW and AIPW) become highly unstable when the consideration set is large, as shown in the simulation results in Section~\ref{sec:simulation}, we exclude these estimators from the empirical analysis. 

% To ensure a fair comparison, the viewer response models and \rzedit{algorithm choice model} employed in the debiased estimator are also leveraged in the benchmark methods where applicable, including the construction of the propensities of the GIPW and GAIPW estimators.

\begin{figure}
     \centering
        \caption{Treatment Effect Estimation}
        \vspace{3mm}
    \begin{subfigure}[b]{0.32\textwidth}
        \caption{\small Outcome 1}
        \includegraphics[width=\textwidth]{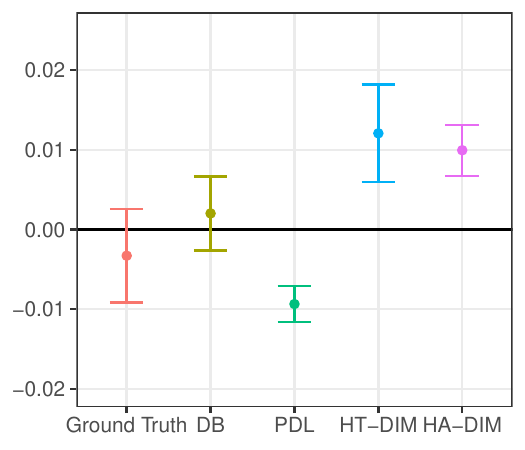}
    \end{subfigure}
    \begin{subfigure}[b]{0.33\textwidth}
            \caption{\small Outcome 2}
        \includegraphics[width=\textwidth]{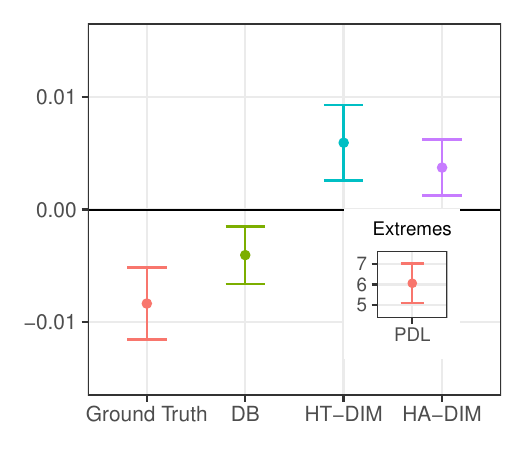}
    \end{subfigure}
    \begin{subfigure}[b]{0.32\textwidth}
        \caption{\small Outcome 3}
        \includegraphics[width=\textwidth]{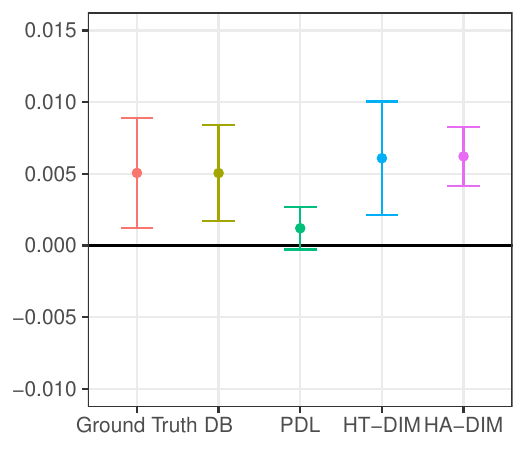}
    \end{subfigure}
    \label{fig:empirical}
      \parbox{0.95\textwidth}{
     \footnotesize \textit{Note:} ``Ground Truth'' is estimated using the double-sided experiment, ``DB'' is the proposed debiased estimator, ``PDL'' is the pure deep learning estimator, ``HT-DIM'' is the Horvitz-Thompson DIM estimator, and ``HA-DIM'' is the Hájek DIM estimator. }
\end{figure}

% \begin{figure}[htbp]
%     \centering

%     \begin{subfigure}[b]{0.48\textwidth}
%         \centering
%         \includegraphics[scale=0.4]{graphs/MS_v2/new_result_scaled_0907.jpg}
%         \caption{Scaled ATE estimation}
%         \label{fig:empirical_scaled}
%     \end{subfigure}
%     \hfill
%     \begin{subfigure}[b]{0.48\textwidth}
%         \centering
%         \includegraphics[scale=0.4]{graphs/MS_v2/new_result_unscaled_0907.jpg}
%         \caption{Unscaled ATE estimation}
%         \label{fig:empirical_unscaled}
%     \end{subfigure}

%     \vspace{2mm}
%     \caption{ATE estimation for three metrics: (a) scaled and (b) unscaled results.}
%     \label{fig:empirical_combined}
% \end{figure}

Figure~\ref{fig:empirical} presents the estimated treatment effects, with each panel corresponding to one of the three outcomes of interest. The solid black lines at zero indicate the natural benchmark for assessing whether the new algorithm should be adopted. For the proposed and benchmark estimators, we evaluate how closely their estimated treatment effects align with the ground-truth measures.\footnote{While the double-sided design eliminates interference between treated and control items, its sub-universes contain only one-third of the market. Since the treatment effect may vary with the market size, the double-sided estimates should be interpreted as our best available benchmark rather than an exact oracle for the full-platform global treatment effect.} 

The practical value and economic significance of our debiased (DB) estimator are most starkly illustrated by the results for outcome 2. For this metric, the ground truth measure from the double-sided experiment reveals a statistically significant negative treatment effect. Our DB estimator successfully recovers this result with comparable magnitude and significance, correctly identifying the new algorithm as the worse option. In sharp contrast, all benchmark estimators fail. Both HT-DIM and HA-DIM report statistically significant positive effects, while the pure deep learning estimator is wildly inaccurate. Relying on these estimators would have led to an incorrect managerial decision regarding which algorithm to adopt. This result highlights the most critical danger of using naive estimators: the potential for reaching directionally wrong conclusions.

The unreliability of the benchmark estimators is also apparent in the other two outcomes. For outcome 1, the ground-truth measure indicates a null effect, and the proposed debiased estimator correctly cannot reject the null hypothesis. In contrast, both HT-DIM and HA-DIM estimators both report statistically significant positive effects. The pure deep learning estimator reports a negative and significant effect. While not as severe as getting the sign wrong, this would still lead the platform to dedicate resources to a costly rollout for a change with no actual benefit. 

For outcome 3, while the benchmarks happen to align with the positive ground truth, their demonstrated failure in the other cases makes this result untrustworthy. The platform should not rely on an estimator that is correct only by chance. Across all three key metrics, the proposed debiased estimator is the only approach that consistently provides reliable results aligned with the ground truth, making it a trustworthy guide for high-stakes business decisions. In contrast, all other estimators yield confidence intervals that fail to overlap with the ground truth in at least two of the three outcomes.

Comparing across the three outcome measures, treatment effect on outcome 2 is particularly challenging to estimate. All benchmark methods perform poorly in this case, with the pure deep learning estimator yielding extreme estimates and the two DIM estimators producing effects with the wrong sign. Although it is difficult to pinpoint the exact source of difficulty, outcome~2 exhibits a highly skewed distribution with a large mass at zero,\footnote{The Fisher–Pearson coefficients of skewness for the three outcomes are 10.46, 21.21, and 11.09, respectively.} which generally makes the estimation of nuisance components more difficult. The debiased estimator mitigates the issue through its Neyman orthogonality property, a feature not shared by these benchmark estimators.

\section{Conclusion}

Creator-side experiments are important for platforms to evaluate creator-focused promotional algorithm updates. This paper demonstrates that algorithmic interference from competition between content creators makes standard DIM estimators dangerously misleading and leads to biased treatment effect estimation. We show empirically that relying on standard DIM estimators can lead to the wrong business decision by deploying the worse algorithm.

To address this challenge, we develop a new approach grounded in the Double/Debiased Machine Learning framework that provides a reliable estimate using data from a standard creator-side experiment. Our method explicitly models the interference pathway with a semi-parametric choice model and uses a debiased estimator to ensure valid inference. We validate our estimator against ground-truth estimates obtained from a costly double-sided experimental design, showing that it successfully recovers the true treatment effect where common benchmarks fail.

This paper's contributions are therefore both practical and theoretical. For practitioners, we offer a trustworthy tool to make better decisions about their algorithms. For academics, our key methodological contribution is the extension of the DML framework to handle the non-i.i.d., correlated data. While our analysis focuses on short-run effects in a re-ranking context, future work could extend this framework to model ranked slate listing or long-term ecosystem effects. More broadly, this framework demonstrates how combining structural and machine learning components in semi-parametric models can address complex interference and identification challenges in modern digital marketplaces.

\newpage

\bibliographystyle{informs2014}
\bibliography{references}

\newpage

\begin{APPENDICES}

\section{Proof of Theorem \ref{thm:dim_bias} (Bias of  DIM Estimators)}
\label{appendix:dim_bias_proof}

We adopt the following condition to control the extent of interference, in line with prior work on interference \citep{viviano2019policy, savje2021average}.  

\begin{assumption}[Controlled Item Appearance]  
\label{assump:item_appearance}  
Let $n$ be the sample size. Each item appears in at most $d_n$  consideration sets, with $\bbE[d_n^3]=o(n)$, where the expectation is taken with respect to the viewer population.  
\end{assumption}  

Let $n$ be the number of viewer queries and $m=|\calC|$ be the number of content items. We assume $m=O(n)$ without loss of generality since at most $nK$ contents are considered.
Let $q$ be the probability of treatment assignment.  
For a viewer $v$, we use $\vv{c}_v=(c_{v,c})$ and $\vv{w}_v=(w_{v,c})$ to denote the content consideration set and the treatment statuses there in, and use $\vv{p}_v = (p_{v,c})$ to denote the exposure probabilities. Let $e^*(q\mid v)$ be the probability of a viewer $v$ seeing treated content under the creator randomization with treatment assignment probability $q$. 
Mathematically, we have
\begin{equation}
    e^*(q\mid v) = \bbE_{\bfw\sim \calB(q)}[w_{k^*}\mid v].
\end{equation}
Similarly, define the population level probability of a viewer seeing a treated content item under the design as:
\begin{equation}
    e^*(q) = \bbE_{v\sim \bbP_v}\bb{e^*(q\mid v) },
\end{equation}
where $\bbP_v$ is the total viewer distribution.

Next, define the distributions of viewers exposed to treated and control content items under the randomization as $\bbP_{v}^{1,q}$ and $\bbP_{v}^{0,q}$ respectively, with:
\begin{equation}
        \frac{\diff\bbP_v^{1,q}}{\diff v} =   \frac{\diff\bbP_v}{\diff v}\cdot \frac{e^*(q\mid v)}{e^*(q)}\quad \mbox{and}\quad      \frac{\diff\bbP_v^{0,q}}{\diff v} =  \frac{\diff\bbP_v}{\diff v}\cdot\frac{1-e^*(q\mid v)}{1-e^*(q)}.
\end{equation}

\subsection{\texorpdfstring{Convergence of Horvitz-Thompson estimator ( $\hat{\tau}^{HT-DIM}_n $)}{HT estimator}}
Given $n$ samples $\{(V_i, \vv{C_i}, \vv{W_i}, k_i^*, Y_i )\}_{i=1}^n$ collected from a creator-side randomization experiment, with treatment probability $q$, 
 the Horvitz-Thompson estimator ($\hat{\tau}^{HT-DIM}_n $) is defined as:
 \begin{equation}
          \hat{\tau}^{HT-DIM}_n := \frac{\sum_{i=1}^n W_{i,k_i^*} Y_i}{nq} - \frac{\sum_{i=1}^n (1-W_{i,k_i^*})Y_i}{n(1-q)}\tag{\ref{eq:HT-DIM}}.
 \end{equation}

We first show that $\bbE[\hat{\tau}^{HT-DIM}_n] = \tau^{HT}$, where $\tau^{HT}$ is defined as:
    \begin{align}
     \tag{\ref{eq:ht_dim_convergence}}
      \quad \tau^{HT} := \frac{e^*(q)}{q}  \sum_{c\in\calC}\bbE_{v\sim \bbP_{v}^{1,q}, \bfw\sim \calB(q)}  [r_c\mid w_{k^*}=1]  
 &- \frac{1-e^*(q)}{1-q}\sum_{c\in\calC}\bbE_{v\sim \bbP_{v}^{0,q}, \bfw\sim \calB(q)}[r_c\mid w_{k^*}=0],
    \end{align}
For each sample $(V_i, \vv{C}_i, \vv{W}_i, k_i^*, Y_i)$ collected from the experiment,
we have:
    \begin{align*}
        \bbE\bb{\frac{1}{q} W_{i,k_i^*}Y_i} 
        &= \frac{1}{q}\bbE_{v\sim \bbP_v}\bb{\bbE_{\bfw\sim\calB(q)}\bb{ w_{k^*}y \mid v}}\\
        &= \frac{1}{q}\bbE_{v\sim \bbP_v^{1,q}}\bb{\frac{e^*(q)}{e^*(q|v)} \bbE_{\bfw\sim\calB(q)}\bb{ w_{k^*}y \mid v}}\\
        &= \frac{e^*(q)}{q}\bbE_{v\sim \bbP_v^{1,q}}\bb{\frac{1}{e^*(q|v)} \bbE_{\bfw\sim\calB(q)}\bb{ w_{k^*}y \mid v}}\\
        &= \frac{e^*(q)}{q}\bbE_{v\sim \bbP_v^{1,q}}\bb{\frac{1}{\bbE_{\bfw\sim\calB(q)}\bb{ w_{k^*}\mid v }} \bbE_{\bfw\sim\calB(q)}\bb{ w_{k^*}y\mid v }}\\
        &= \frac{e^*(q)}{q}\bbE_{v\sim \bbP_v^{1,q}}\bb{\frac{1}{\bbE_{\bfw\sim\calB(q)}\bb{ w_{k^*}\mid v }} \bbE_{\bfw\sim\calB(q)}\bb{ w_{k^*}\mid v } \bbE_{\bfw\sim\calB(q)}\bb{ y\mid w_{k^*}=1,v } }\\
        &= \frac{e^*(q)}{q}\bbE_{v\sim \bbP_v^{1,q}}\bb{ \bbE_{\bfw\sim\calB(q)}\bb{ y\mid w_{k^*}=1,v } }\\
        &= \frac{e^*(q)}{q}\bbE_{v\sim \bbP_{v}^{1,q},\bfw\sim\calB(q)}\bb{ y\mid w_{k^*}=1 } \\
        & = \frac{e^*(q)}{q}  \sum_{c\in\calC}\bbE_{v\sim \bbP_{v}^{1,q}, \bfw\sim \calB(q)}  [r_c\mid w_{k^*}=1], 
    \end{align*}
where the last line uses the definition of $r_c$ in Section \ref{sec:recsys} such that $r_c = y$ if content $c$ is exposed to viewer $v$, otherwise $r_c=0$.

Similarly, we have
    \begin{equation}
         \bbE\bb{\frac{1}{1-q} (1-W_{i,k_i^*})Y_i} =  \frac{1-e^*(q)}{1-q}\sum_{c\in\calC}\bbE_{v\sim \bbP_{v}^{0,q}, \bfw\sim \calB(q)}[r_c\mid w_{k^*}=0].
    \end{equation}
Recall that $\hat{\tau}^{HT-DIM}_n = \frac{1}{n}\sum_{i=1}^n \frac{1}{q} W_{i,k_i^*}Y_i - \frac{1}{1-q} (1-W_{i,k_i^*})Y_i$. Therefore we have  $\bbE[\hat{\tau}^{HT-DIM}_n] = \tau^{HT}$.

Next, we characterize the variance of $\hat{\tau}^{HT-DIM}_n$. By the law of total variance, we have
  \[
  \Var(\hat{\tau}^{HT-DIM}_n) = \bbE\bb{\Var\bp{\hat{\tau}^{HT-DIM}_n\mid \{(V_i,\vv{C}_i)\}_{i=1}^n}} + \Var\bp{\bbE\bb{\hat{\tau}^{HT-DIM}_n\mid \{(V_i,\vv{C}_i)\}_{i=1}^n}}.
  \]
  Conditioning on $\{(V_i,\vv{C}_i)\}_{i=1}^n$,  for each content $c$,
let $w_c$ denote its treatment status,  $d_c$ denote the number of consideration set including content $c$, and $R_c :=\frac{1}{n} \sum_{i=1}^n \one\bc{c= C_{i,k^*_i}}Y_i$ to denote the   average viewer-outcome for content $c$.  Denote $d = \max_{c}d_c$.
  We have: 
  \begin{align*}
      &
      \Var\bp{\frac{w_cR_c}{q} - \frac{(1-w_c)R_c}{1-q} \mid \{(V_i,\vv{C}_i)\}_{i=1}^n}\\
      &= \frac{\Var(w_cR_c\mid \{(V_i,\vv{C}_i))\}_{i=1}^n}{q^2} + \frac{\Var((1-w_c)R_c \mid \{(V_i,\vv{C}_i)\}_{i=1}^n)}{(1-q)^2} - \frac{2\Cov(w_cR_c, (1-w_c)R_c \mid \{(V_i,\vv{C}_i)\}_{i=1}^n)}{q(1-q) }\\
      &=\frac{\bbE[w_cR_c^2 \mid \{(V_i,\vv{C}_i)\}_{i=1}^n] -\bbE[w_cR_c \mid \{(V_i,\vv{C}_i)\}_{i=1}^n]^2 }{q^2} + \frac{\bbE[(1-w_c)R_c^2 \mid \{(V_i,\vv{C}_i)\}_{i=1}^n]-\bbE[(1-w_c)R_c]^2}{(1-q)^2 \mid \{(V_i,\vv{C}_i)\}_{i=1}^n}\\
      &\quad\quad + \frac{2\bbE[w_cR_c \mid \{(V_i,\vv{C}_i)\}_{i=1}^n]\bbE[(1-w_c)R_c \mid \{(V_i,\vv{C}_i)\}_{i=1}^n]}{q(1-q)}\\
      &=\frac{\bbE[R_c^2\mid w_c=1,\{(V_i,\vv{C}_i)\}_{i=1}^n] }{q}  -\bbE[R_c|w_c=1,\{(V_i,\vv{C}_i)\}_{i=1}^n]^2 \\
      &\quad\quad
      + \frac{\bbE[R_c^2\mid w_c=0,\{(V_i,\vv{C}_i)\}_{i=1}^n]}{1-q} - 
    \bbE[R_c\mid w_c=0,\{(V_i,\vv{C}_i)\}_{i=1}^n]^2
      \\
      &\quad\quad+ 2\bbE[R_c\mid w_c=1,\{(V_i,\vv{C}_i)\}_{i=1}^n]\bbE[R_c\mid w_c=0,\{(V_i,\vv{C}_i)\}_{i=1}^n]\\
      &=\frac{\bbE[R_c^2\mid w_c=1,\{(V_i,\vv{C}_i)\}_{i=1}^n] }{q} 
      + \frac{\bbE[R_c^2\mid w_c=0,\{(V_i,\vv{C}_i)\}_{i=1}^n]}{1-q} \\
      &\quad\quad- \bp{\bbE[R_c\mid w_c=1,\{(V_i,\vv{C}_i)\}_{i=1}^n]-\bbE[R_c\mid w_c=0,\{(V_i,\vv{C}_i)\}_{i=1}^n]}^2.
  \end{align*}
  Recall that $R_c :=\frac{1}{n} \sum_{i=1}^n \one\bc{c= C_{i,k^*_i}}Y_i$. By the boundedness of $Y_i$, we have that $R_c=O(\frac{d_c}{n})$, and thus:
  \[
  \Var\bp{\frac{w_cR_c}{q} - \frac{(1-w_c)R_c}{1-q}\mid \{(V_i,\vv{C}_i)\}_{i=1}^n} =O\bp{\frac{d_c^2}{n^2}} = O\bp{\frac{d^2}{n^2}}.
  \]
Also conditioning $\{(V_i,\vv{C}_i)\}_{i=1}^n$,  define $I_{c_1,c_2}=1$ if item $c_1$ and item $c_2$ are present at one consideration set;\footnote{By convention, we say  that $c_1$ and $c_2$ are in the same consideration set if both refer to the same item.} otherwise set $I_{c_1,c_2}=0$. Note that if $I_{c_1,c_2}=0$, there is no interference among the items $c_1~\&~c_2$ and thus $R_{c_1}$ and $R_{c_2}$ are independent. 
\begin{align*}
    &\Cov\bp{
    \frac{w_{c_1}R_{c_1}}{q} - \frac{(1-w_{c_1})R_{c_1}}{1-q}, \frac{w_{c_2}R_{c_2}}{q} - \frac{(1-w_{c_2})R_{c_2}}{1-q}
    ~\Big|~ \{(V_i,\vv{C}_i)\}_{i=1}^n} \\
    = & \left\{
    \begin{aligned}
       & 0 &&\mbox{if}\quad I_{c_1,c_2}=0,\\
       &O\Bigg(\sqrt{\Var\bp{\frac{w_{c_1}R_{c_1}}{q} - \frac{(1-w_{c_1})R_{c_1}}{1-q}\mid  \{(V_i,\vv{C}_i)\}_{i=1}^n}}&& \mbox{o.w.}\\
       &\quad \quad \times \sqrt{\Var\bp{\frac{w_{c_2}R_{c_2}}{q} - \frac{(1-w_{c_2})R_{c_2}}{1-q}\mid  \{(V_i,\vv{C}_i)\}_{i=1}^n}}\Bigg) 
    \end{aligned}
    \right.
\end{align*}
Together, we have
\begin{align*}
    \Var\bp{\hat{\tau}^{HT-DIM}_n\mid \{(V_i,\vv{C}_i)\}_{i=1}^n} \lesssim &\sum_{c\in\calC}^m \Var\bp{\frac{w_cR_c}{q} - \frac{(1-w_c)R_c}{1-q} \mid \{(V_i,\vv{C}_i)\}_{i=1}^n} \\
    &\quad\quad + \sum_{c_1\neq c_2\in\calC}\Cov\bp{
    \frac{w_{c_1}R_{c_1}}{q} - \frac{(1-w_{c_1})R_{c_1}}{1-q}, \frac{w_{c_2}R_{c_2}}{q} - \frac{(1-w_{c_2})R_{c_2}}{1-q}
    \mid \{(V_i,\vv{C}_i)\}_{i=1}^n} \\
   = & O\bp{\frac{d^2}{n^2}} + \sum_{c_1\neq c_2\in\calC} I_{c_1,c_2}  O\bp{\frac{d^2}{n^2}}\\
   = &  O_p\bp{\frac{md^3}{n^2}} =   O_p\bp{\frac{d^3}{n}}.
\end{align*}
Thus, 
\begin{equation}
\label{eq:total_law_variance_1}
\bbE\bb{ \Var\bp{\hat{\tau}^{HT-DIM}_n\mid \{(V_i,\vv{C}_i)\}_{i=1}^n}}=O\bp{\frac{\bbE[d^3]}{n}}.
\end{equation}
On the other hand, we have 
\begin{align*}
     \bbE\bb{\hat{\tau}^{HT-DIM}_n\mid \{(V_i,\vv{C}_i)\}_{i=1}^n}&= \sum_{c\in\calC} \bbE\bb{\frac{1}{n}\sum_{i=1}^n \frac{\one\bc{c=C_{i,k_i^*}}w_c}{q}Y_i
     -
     \frac{1}{n}\sum_{i=1}^n \frac{ \one\bc{c=C_{i,k_i^*}}(1-w_c)}{1-q}Y_i
     | \{(V_i,\vv{C}_i)\}_{i=1}^n}\\
     &= \frac{1}{n}\sum_{i=1}^n\sum_{c\in\calC} \bc{\bbE\bb{ \frac{\one\bc{c=C_{i,k_i^*}}w_c}{q}Y_i| \{(V_i,\vv{C}_i)\}_{i=1}^n} - \bbE\bb{ \frac{ \one\bc{c=C_{i,k_i^*}}(1-w_c)}{1-q}Y_i|  \{(V_i,\vv{C}_i)\}_{i=1}^n}}\\
     &= \frac{1}{n}\sum_{i=1}^nZ_i,\\
     \mbox{where}\quad Z_i &= \bbE\bb{ \bp{\frac{W_{i,k_i^*}}{q}-\frac{1-W_{i,k_i^*}}{1-q}}Y_i\mid \{(V_i,\vv{C}_i)\}_{i=1}^n}= \bbE\bb{ \bp{\frac{W_{i,k_i^*}}{q}-\frac{1-W_{i,k_i^*}}{1-q}}Y_i\mid (V_i,\vv{C}_i)}.
\end{align*}
Note that the randomness in $Z_i$ comes solely from $(V_i,\vv{C}_i)$, which are i.i.d.~across $i$.
Thus,
\begin{align}
    \Var\bp{ \bbE\bb{\hat{\tau}^{HT-DIM}_n\mid \{(V_i,\vv{C}_i)\}_{i=1}^n}} & = \Var\bp{\frac{1}{n}\sum_{i=1}^nZ_i} = \frac{\sum_{i=1}^n \Var(Z_i)}{n^2}=O\bp{\frac{1}{n}}.\label{eq:total_law_variance_2}
\end{align}
Combining Equation~\eqref{eq:total_law_variance_1} \& Equation~\eqref{eq:total_law_variance_2}, we have
\begin{align*}
\Var\bp{\hat{\tau}^{HT-DIM}_n} &=\bbE\bb{ \Var\bp{\hat{\tau}^{DIM}_n\mid \{(V_i,\vv{C}_i)\}_{i=1}^n}} +  \Var\bp{ \bbE\bb{\hat{\tau}^{DIM}_n\mid \{(V_i,\vv{C}_i)\}_{i=1}^n}}\\
&=O\bp{\frac{\bbE[d^3]}{n}} + O\bp{\frac{1}{n}} =O\bp{\frac{\bbE[d^3]}{n}}= o(1),   
\end{align*}
where the last equation is by Assumption \ref{assump:item_appearance}. Recall that we have proved $\bbE[\hat{\tau}^{HT-DIM}_n]$ equals to $\tau^{HT}$; then by Markov inequality, 
we have $\hat{\tau}^{HT-DIM}_n\xrightarrow{p}\tau^{HT}$.

\subsection{\texorpdfstring{Convergence of Hájek estimator ($\hat{\tau}^{HA-DIM}_n$)}{Hajek estimator}}
Given $n$ samples $\{(V_i, \vv{C_i}, \vv{W_i}, k_i^*, Y_i )\}_{i=1}^n$ collected from a creator-side randomization experiment, with treatment probability $q$, 
 the Hájek estimator ( $\hat{\tau}^{HA-DIM}_n $) is defined as:
 \begin{equation}
          \hat{\tau}^{HA-DIM}_n := \frac{\sum_{i=1}^nW_{i,k_i^*}Y_i}{\sum_{i=1}^nW_{i,k_i^*}} - \frac{\sum_{i=1}^n(1-W_{i,k_i^*})Y_i}{\sum_{i=1}^n(1-W_{i,k_i^*})},\tag{\ref{eq:HA-DIM}}.
 \end{equation}

We here show that $\hat{\tau}^{HA-DIM}_n\xrightarrow{p} \tau^{HA}$, where $\tau^{HA}$ is defined as:
    \begin{align}
     \tag{\ref{eq:ha_dim_convergence}}
      \quad \tau^{HA} :=  \sum_{c\in\calC}\bbE_{v\sim \bbP_{v}^{1,q}, \bfw\sim \calB(q)}  [r_c\mid w_{k^*}=1]  
 &- \sum_{c\in\calC}\bbE_{v\sim \bbP_{v}^{0,q}, \bfw\sim \calB(q)}[r_c\mid w_{k^*}=0],
    \end{align}

 By Slutsky's theorem, the missing component to show the convergence of $\hat{\tau}_n^{HA-DIM}$ is the convergence of $\frac{1}{n}\sum_{i=1}^n W_{i,k_i^*}$ to $e^*(q)$. Then  with  the convergence of $\hat{\tau}_n^{HT-DIM}\xrightarrow{p} \tau^{HT}$, we finish the proof. 
    
We now show the convergence of $\frac{1}{n}\sum_{i=1}^n W_{i,k}$. Note that by definition of $e^*(q)$, we have:
    \begin{equation*}
        \bbE\bb{\frac{1}{n}\sum_{i=1}^n W_{i,k_i^*}} = e^*(q).
    \end{equation*}
    Now we characterize its variance. 
    \begin{equation}
        \begin{split}
            \Var\bp{\frac{1}{n}\sum_{i=1}^n W_{i,k_i^*}} 
            &=\bbE\bb{\Var\bp{\frac{1}{n}\sum_{i=1}^n W_{i,k_i^*} | \{(V_i,\vv{C}_i)\}_{i=1}^n} }+\Var\bb{\bbE\bp{\frac{1}{n}\sum_{i=1}^n W_{i,k_i^*} | \{(V_i,\vv{C}_i)\}_{i=1}^n} }\\
            &= \frac{1}{n^2}\bbE\bb{\sum_{i\ne j}\Cov\bp{ W_{i,k_i^*}, W_{j,k_j^*}| \{(V_i,\vv{C}_i)\}_{i=1}^n} } + O\bp{\frac{1}{n}}\\
            &= O\bp{\frac{\bbE[d^3]}{n}} + O\bp{\frac{1}{n}} = o(1).
        \end{split}
    \end{equation}
    %where the last equation is by Assumption \ref{assump:item_apprearance}.
    By Markov inequality, we have 
    \begin{equation}
        \frac{1}{n}\sum_{i=1}^n W_{i,k^*}\xrightarrow{p}e^*(q).
    \end{equation}
    Similarly, we  have:
    \begin{equation}
        \frac{1}{n}\sum_{i=1}^n (1-W_{i,k^*})\xrightarrow{p}1-e^*(q).
    \end{equation}
    Note that,
    \begin{align*}
        \hat{\tau}_n^{HA-DIM} &= \frac{n}{\sum_{i=1}^n W_{i,k^*}}\frac{\sum_{i=1}^n W_{i,k^*}Y_i}{n} -  \frac{n}{\sum_{i=1}^n (1-W_{i,k^*})}\frac{\sum_{i=1}^n (1-W_{i,k^*})Y_i}{n} \\
        \hat{\tau}_n^{HT-DIM} &= \frac{\sum_{i=1}^n W_{i,k^*}Y_i}{nq} -  \frac{\sum_{i=1}^n (1-W_{i,k^*})Y_i}{n(1-q)} .
    \end{align*}
    By Slutsky's theorem, we have $  \hat{\tau}_n^{HA-DIM}\xrightarrow{p}\tau^{HA}$.

\newpage

% \section{Proof of Lemma \ref{lemma:vanilla_ipw_estimator} (Variance of IPW estimator)}
% \label{appendix:vanilla_ipw_variance}

% \begin{proof}{Proof of Lemma \ref{lemma:vanilla_ipw_estimator}.}
% By the definition of $\pi_\one$ and $\pi_\zero$,
% \[
% \hat{\tau}_n^{IPW}
% = \frac{1}{n} \sum_{i=1}^n 
%   \left[\frac{\one\{\bfw_i=\one\}}{q^K}
%  - \frac{\one\{\bfw_i=\zero\}}{(1-q)^K}\right] Y_i,
% \]
% and on the events $\{\bfw_i=\one\}$ and $\{\bfw_i=\zero\}$ we have
% $Y_i = Y_i(\one)$ and $Y_i = Y_i(\zero)$, respectively. 
% Conditioning on the potential outcomes and using $\bbP(\bfw_i=\one)=q^K$ and $\bbP(\bfw_i=\zero)=(1-q)^K$,
% \[
% \bbE\left[\hat{\tau}_n^{IPW} \,\big|\, \{Y_i(\one),Y_i(\zero)\}_{i=1}^n\right]
% = \frac{1}{n} \sum_{i=1}^n \{Y_i(\one)-Y_i(\zero)\},
% \]
% so taking expectations yields $\bbE[\hat{\tau}_n^{IPW}]=\tau$.

% For the variance, independence across $i$ implies
% \[
% \Var(\hat{\tau}_n^{IPW})
% = \frac{1}{n}\Var\left(
%   \left[\frac{\one\{\bfw_i=\one\}}{q^K}
%  - \frac{\one\{\bfw_i=\zero\}}{(1-q)^K}\right] Y_i
% \right).
% \]
% Again conditioning on $Y_i(\one),Y_i(\zero)$ and using that the events
% $\{\bfw_i=\one\}$ and $\{\bfw_i=\zero\}$ are disjoint,
% \[
% \bbE\left(
%   \left[\frac{\one\{\bfw_i=\one\}}{q^K}
%  - \frac{\one\{\bfw_i=\zero\}}{(1-q)^K}\right]^2
%   Y_i^2
%   \,\Big|\,
%   Y_i(\one),Y_i(\zero)
% \right)
% =
% \frac{Y_i(\one)^2}{q^K} + \frac{Y_i(\zero)^2}{(1-q)^K},
% \]
% and thus
% \begin{equation}
% \Var(\hat{\tau}_n^{IPW}) = \Theta\left(\frac{q^{-K} + (1-q)^{-K}}{n}\right).
% \end{equation}
% \end{proof}

\section{Explicit Expressions}

\subsection{Counterfactual Policy Value}
\label{appendix:policy_value}
We write out the explicit expression of counterfactual policy values using the algorithm choice model. Given a policy $\pi$, the policy value is
\begin{align*}
 Q(\pi)&=\bbE_{\bfw\sim \pi}\bb{ \sum_{c\in\calC} r(c;w_c, \bfw_{-c}) }=\bbE_{\bfw\sim \pi}\bb{\sum_{c\in\calC}  \bbE_v[ y\bp{v,c; w_{c}, \bfw_{-c}}]}\\
 &= \bbE_v\bb{  \bbE_{\bfw\sim \pi}\bb{\sum_{c\in\calC} y\bp{v,c; w_{c}, \bfw_{-c}}}} = \bbE_{V_i, \vv{C}_i}\bb{  \bbE_{\vv{W}_i\sim \pi}\bb{\sum_{c\in C_i} y\bp{V_i,c; \vv{W}_i}}}\\
& =  \bbE_{V_i,\vv{C}_i, \vv{W_i}\sim \pi}\bb{Y_i(V_i, \vv{C_i};\vv{W_i}) }\\
&= \bbE_{V_i, \vv{C}_i,\vv{W_i}\sim \pi}\bb{\sum_{k=1}^K z(V_i, C_{i,k}) \cdot p_k\bp{ V_i, \vv{C_i}, \vv{W_i};s_0,s_1} }\\
&= \bbE_{V_i, \vv{C_i}, \vv{W_i}\sim \pi}\bb{\sum_{k=1}^K z(V_i, C_{i,k}) \cdot \frac{e^{s_0(V_i, C_{i,k}) +  W_{i,k}\cdot s_1(V_i, C_{i,k})}}{\sum_{k'=1}^K e^{s_0(V_i, C_{i,k'}) +  W_{i,k'} \cdot s_1(V_i, C_{i,k'})}} }.
\end{align*}

\subsection{Debiased Estimator}
\label{appendix:debiased_term}
We write out the explicit debiased estimate $\psi(\cdot)$ for each observation, where  we drop the subscript $i$ and write the notation as $(V,\vv{C},\vv{W},k^*,Y)$ for succinctness. 
For  estimated baseline score function $\hat{s}_0\in\calF_s$, treatment score uplift function $\hat{s}_1\in \calF_s$, and viewer response function $\hat{z}\in\calF_z$,  we have the estimated exposure probability be
\begin{align*}
p_k\bp{ V, \vv{C}, \vv{W}; \hat{s}_0,  \hat{s}_1} &= \frac{e^{ \hat{s}_0(V, C_{k}) +  W_{k} \cdot  \hat{s}_1(V, C_{k})}}{\sum_{k'=1}^K e^{ \hat{s}_0(V, C_{k'}) +  W_{k'} \cdot  \hat{s}_1(V, C_{k'})}}\stackrel{(i)}{=} \frac{e^{ \hat{s}_0(V, C_{k}) -  \hat{s}_0(V, C_{1}) +  W_{k} \cdot  \hat{s}_1(V, C_{k})}}{\sum_{k'=1}^K e^{ \hat{s}_0(V, C_{k'}) -  \hat{s}_0(V, C_{1}) +  W_{k'} \cdot  \hat{s}_1(V, C_{k'})}},
\end{align*}
where the equation (i) is obtained by normalizing both the numerator and denominator by the exponential of the first content item's baseline score. In other words, for any baseline score vector $\bp{\hat{s}_0(V, C_1), \hat{s}_0(V, C_2),\dots, \hat{s}_0(V, C_K)}$, if we replace it by  $\bp{0, \hat{s}_0(V, C_2)-\hat{s}_0(V, C_1), \dots, \hat{s}_0(V, C_K)-\hat{s}_0(V, C_1)}$, we will get the same exposure probability result. 

This implies that, for any nuisance estimates $(\hat{s}_0, \hat{s}_1,\hat{z})$,  the value $\mu(V,\vv{C};\hat{s}_0, \hat{s}_1,\hat{z})$ can be fully recovered by the vectors 
\begin{itemize}
\item $\vv{\hat{S}_0}=(\hat{S}_{0,2}, \dots, \hat{S}_{0,K}) \in \bbR^{K-1}$ with $\hat{S}_{0,k}=\hat{s}_0(V, C_k)-\hat{s}_0(V, C_{1})$;

\item  $\vv{\hat{S}_1}=(\hat{S}_{1,1}, \dots, \hat{S}_{1,K}) \in \bbR^K$ with $\hat{S}_{1,k}=\hat{s}_1(V, C_k)$;

\item  $\vv{\hat{Z}}=(\hat{Z}_{1,1}, \dots, \hat{Z}_{1,K}) \in \bbR^K $ with $\hat{Z}_{1,k}=\hat{z}(V, C_k)$.

\end{itemize}
 We thus abuse notations and write $\mu(V,\vv{C};\tilde{s}_0, \tilde{s}_1,\tilde{z})$  as $\mu\bp{\vv{\hat{S}_0}, \vv{\hat{S}_1}, \vv{\hat{Z}}}$. 

The bias of $\mu\bp{\vv{\hat{S}_0}, \vv{\hat{S}_1}, \vv{\hat{Z}}}$ comes from the deviation of $\bp{\vv{\hat{S}_0}, \vv{\hat{S}_1}, \vv{\hat{Z}}}$ to the true vectors $\bp{\vv{S_0}, \vv{S_1}, \vv{Z}}$ that are defined similarly  under  the true model $(s_0, s_1, z)$. 
 
 We next follow the double machine learning literature and use the outcome $(k^*, Y)$ to  correct the bias of $\mu\bp{\vv{\hat{S}_0}, \vv{\hat{S}_1}, \vv{\hat{Z}}}$ due to the bias of $\bp{\vv{\hat{S}_0}, \vv{\hat{S}_1}, \vv{\hat{Z}}}$ approximating the true $\bp{\vv{S_0}, \vv{S_1}, \vv{Z}}$. 
 
 Under the estimates  $\bp{\vv{\hat{S}_0}, \vv{\hat{S}_1}, \vv{\hat{Z}}}$, we reload the loss function notation and write it as:
 \begin{align*}
 \ell\bp{ \vv{W}, k^*, Y; \vv{\hat{S}_0}, \vv{\hat{S}_1}, \vv{\hat{Z}}} &= \ell_1\bp{\vv{W}, k^*; \vv{\hat{S}_0}, \vv{\hat{S}_1}} + \ell_2\bp{k^*, Y; \vv{\hat{Z}}},
 \end{align*}
 where 
 \begin{equation*}
\ell_1\bp{\vv{W}, k^*; \vv{\hat{S}_0}, \vv{\hat{S}_1}} = \left\{
 \begin{aligned}
& -W_{1}\hat{S}_{1,1} +\log\bp{e^{W_{1}\hat{S}_{1,1}} + \sum_{i=2}^K e^{\hat{S}_{0,k}+W_{k}\hat{S}_{1,k}} } \quad &&\mbox{if }k^*=1,\\
& -\bp{\hat{S}_{0,k}+W_{k}\hat{S}_{1,k}} +\log\bp{e^{W_{1}\hat{S}_{1,1}} + \sum_{i=2}^K e^{\hat{S}_{0,k}+W_{k}\hat{S}_{1,k}} } \quad &&
 \mbox{o.w.;}
 \end{aligned}
 \right.
 \end{equation*}
 and 
 \begin{equation*}
\ell_2\bp{k^*, Y; \vv{\hat{Z}}}, = \bp{\hat{Z}_{k^*}-Y}^2.
 \end{equation*}

 We are now ready to introduce the debiased term $\psi$:
\begin{align*}
\psi\bp{ \vv{W}, k^*, Y; \vv{\hat{S}_0}, \vv{\hat{S}_1}, \vv{\hat{Z}},\hat{H} } =~\mu\bp{\vv{\hat{S}_0}, \vv{\hat{S}_1}, \vv{\hat{Z}}} - \nabla\mu^T \hat{H}^{-1} \nabla\ell.
\end{align*}
Above, $\nabla\mu$ and $\nabla\ell$ are gradients of $\mu$ and $\ell$ with respect to the nuisance estimates $\bp{\vv{\hat{S}_0}, \vv{\hat{S}_1}, \vv{\hat{Z}}}$, and $\hat{H}$ estimates the expected Hessian of $\ell$ regarding $\bp{\vv{\hat{S}_0}, \vv{\hat{S}_1}, \vv{\hat{Z}}}$, where the expectation is taken with respect to the treatments $\vv{W}$ that are assigned following the specified creator-side randomization. The explicit expressions for these derivatives are deferred to Appendix \ref{appendix:gradients}.

\subsection{Expressions of Gradients and Hessian}
\label{appendix:gradients}
\subsubsection{\texorpdfstring{Gradient of $\mu$.}{Gradient}} 
\label{appendix:gradient_of_mu}
We have
\begin{align*}
\nabla \mu =\bp{ \frac{\partial \mu}{\partial \vv{\hat{S}_{0}}},  \frac{\partial \mu}{\partial \vv{\hat{S}_{1}}}, \frac{\partial \mu}{\partial \vv{\hat{Z}}}}^T = \bp{
    \frac{\partial \mu}{\partial \hat{S}_{0,2}}, \dots,   \frac{\partial \mu}{\partial \hat{S}_{0,K}},
      \frac{\partial \mu}{\partial \hat{S}_{1,1}}, \dots,   \frac{\partial \mu}{\partial \hat{S}_{1,K}},
     \frac{\partial \mu}{\partial \hat{Z}_1}, \dots, \frac{\partial \mu}{\partial \hat{Z}_K}
     }^T,
\end{align*}
where
\begin{itemize}
    \item for each $k=2,\dots, K$,
    \begin{align*}
        \frac{\partial\mu }{\partial \hat{S}_{0,k}} = P(k^*=k; \vv{\hat{S}_{0}}, \vv{\hat{S}_{1}}, \vv{W}\equiv 1)\bc{\hat{Z}_{k}-\bbE[Y\mid \vv{\hat{S}_{0}}, \vv{\hat{S}_{1}}, \vv{\hat{Z}},    \vv{W}\equiv 1]} \\
        -  P(k^*=k; \vv{\hat{S}_{0}}, \vv{\hat{S}_{1}}, \vv{W}\equiv 0)\bc{\hat{Z}_{k}-\bbE[Y\mid \vv{\hat{S}_{0}}, \vv{\hat{S}_{1}}, \vv{\hat{Z}},    \vv{W}\equiv 0]} .
    \end{align*}
    \item for each $k=1,\dots, K$,
    \begin{align*}
   \frac{\partial\mu }{\partial \hat{S}_{1,k}}  &=P(k^*=k; \vv{\hat{S}_{0}}, \vv{\hat{S}_{1}}, \vv{W}\equiv 1)
   \bc{
   \hat{Z}_{k}-\bbE[Y\mid \vv{\hat{S}_{0}}, \vv{\hat{S}_{1}}, \vv{\hat{Z}},    \vv{W}\equiv 1]
   },\\
    \frac{\partial\mu }{\partial \hat{Z}_{k}} &= P(k^*=k; \vv{\hat{S}_{0}}, \vv{\hat{S}_{1}}, \vv{W}\equiv 1) -P(k^*=k; \vv{\hat{S}_{0}}, \vv{\hat{S}_{1}}, \vv{W}\equiv 0).
\end{align*}
\end{itemize}
Above, $\langle\cdot, \cdot\rangle$ denotes inner product between two vectors.

\subsubsection{\texorpdfstring{Gradient of $\ell$.}{Gradient}}
\label{appendix:gradient_of_ell}
We have
\begin{align*}
  \nabla \ell = \bp{ \frac{\partial \ell_1}{\partial \vv{\hat{S}_{0}}},  \frac{\partial \ell_1}{\partial \vv{\hat{S}_{1}}}, \frac{\partial \ell_2}{\partial \vv{\hat{Z}}}}^T = \bp{
    \frac{\partial \ell_1}{\partial \hat{S}_{0,2}}, \dots,   \frac{\partial \ell_1}{\partial \hat{S}_{0,K}},
      \frac{\partial \ell_1}{\partial \hat{S}_{1,1}}, \dots,   \frac{\partial \ell_1}{\partial \hat{S}_{1,K}},
     \frac{\partial \ell_2}{\partial \hat{Z}_1}, \dots, \frac{\partial \ell_2}{\partial \hat{Z}_K}
     }^T,
\end{align*}
where

\begin{itemize}
    \item for each $k=2,\dots, K$:
    \begin{align*}
    \frac{\partial\ell_1 }{\partial \hat{S}_{0,k}}  &= P(k^*=k; \vv{\hat{S}_{0}}, \vv{\hat{S}_{1}}, \vv{W}) -\bbI[k^*=k].    
    \end{align*}

    \item for each $k=1,\dots, K$:
    \begin{align*}
    \frac{\partial\ell_1 }{\partial \hat{S}_{1,k}}  &= \bbI\bb{ W_k=1} \bp{P(k^*=k; \vv{\hat{S}_{0}}, \vv{\hat{S}_{1}}, \vv{W}) -\bbI[k^*=k]},\\
   \frac{\partial \ell_2}{\partial \hat{Z}_{k}} & =\bbI[k^*=k]\bp{\hat{Z}_{k}-Y}. 
    \end{align*}
\end{itemize}

\subsubsection{\texorpdfstring{Hessian of  $\ell$.}{Hessian}} \label{appendix:hessian}We have
\begin{align*}
 \nabla^2\ell =\begin{pmatrix}
        \frac{\partial^2 \ell_{1}}{\partial  \vv{\hat{S}_{0}}^2} &  
        \frac{\partial^2 \ell_{1}}{\partial \vv{\hat{S}_{0}} \partial \vv{\hat{S}_{1}} }   & 0\\
        \frac{\partial^2 \ell_{1}}{\partial \vv{\hat{S}_{1}} \partial \vv{\hat{S}_{0}} }  &   \frac{\partial^2 \ell_{1}}{\partial  \vv{\hat{S}_{1}}^2} & 0\\
      0&  0 &   \frac{\partial^2 \ell_2}{\partial \hat{Z}^2}
    \end{pmatrix},
\end{align*}
where the Hessian of loss function $\ell_1 $ follows;
 \begin{align*}
  \frac{\partial^2 \ell_1}{\partial \hat{S}_{0,k}^2}   &= P(k^*=k; \vv{\hat{S}_{0}}, \vv{\hat{S}_{1}}, \vv{W})\bp{1- P(k^*=k; \vv{\hat{S}_{0}}, \vv{\hat{S}_{1}}, \vv{W})}, 
   &&k \in \{2:K\} \\  
     \frac{\partial^2 \ell_1}{\partial \hat{S}_{1,k}^2}    &= W_{k}P(k^*=k; \vv{\hat{S}_{0}}, \vv{\hat{S}_{1}}, \vv{W})\bp{1- P(k^*=k; \vv{\hat{S}_{0}}, \vv{\hat{S}_{1}}, \vv{W})},  &&k \in \{1:K\} \\
      \frac{\partial^2 \ell_1}{\partial  \hat{S}_{0,k} \partial \hat{S}_{1,k}}   &= W_{k}P(k^*=k; \vv{\hat{S}_{0}}, \vv{\hat{S}_{1}}, \vv{W})\bp{1- P(k^*=k; \vv{\hat{S}_{0}}, \vv{\hat{S}_{1}}, \vv{W})}, &&k \in \{2:K\}  \\
       \frac{\partial^2 \ell_1}{\partial \hat{S}_{0,k_1} \partial\hat{S}_{0,k_2}}    &= -P(k^*=k_1; \vv{\hat{S}_{0}}, \vv{\hat{S}_{1}}, \vv{W})P(k^*=k_2; \vv{\hat{S}_{0}}, \vv{\hat{S}_{1}}, \vv{W}), && k_1\neq k_2 \in\{2:K\}\\
         \frac{\partial^2 \ell_1}{\partial \hat{S}_{0,k_1} \partial \hat{S}_{1,k_2}}    &= -W_{k_2}P(k^*=k_1; \vv{\hat{S}_{0}}, \vv{\hat{S}_{1}}, \vv{W})P(k^*=k_2; \vv{\hat{S}_{0}}, \vv{\hat{S}_{1}}, \vv{W}), && k_1\neq k_2, k_1 \in\{2:K\}, k_2\in\{1:K\}\\
         \frac{\partial^2 \ell_1}{\partial \hat{S}_{1,k_1} \partial
        \hat{S}_{1,k_2}
        }    &= -W_{k_1}W_{k_2} P(k^*=k_1; \vv{\hat{S}_{0}}, \vv{\hat{S}_{1}}, \vv{W})P(k^*=k_2; \vv{\hat{S}_{0}}, \vv{\hat{S}_{1}}, \vv{W})
        , && k_1\neq k_2 \in\{1:K\}.
    \end{align*}
and  the Hessian of loss function $\ell_2$ follows:
\begin{align*}
   &  \frac{\partial \ell_2^2}{\partial \hat{Z}_k^2} = \bbI[k^*=k], &&k \in \{1:K\} \\
   & \frac{\partial \ell_2^2}{\partial \hat{Z}_{k_1}\partial  \hat{Z}_{k_2}}=0, && k_1\neq k_2\in\{1:K\}.
\end{align*}

\newpage

\section{Supporting Results for Estimating Nuisance Components}
\label{appendix:nuisance_estimation}

\subsection{Overlap and Positive Exposure}
\label{appendix:positive_exposure}

\begin{assumption}[Overlap]
    \label{assump:exposure_overlap}
    There exists a universal constant $C$ such that the true score functions  are bounded: $\|s_0\|_{\infty}\leq C$ and  $\|s_1\|_{\infty}\leq C$.
\end{assumption}

This assumption plays the same role as the overlap condition in standard causal inference. It ensures that each item in the consideration set has a positive probability of being exposed under both treatment and control algorithms. This property guarantees that counterfactual exposure probabilities and nuisance functions are well defined.

\begin{lemma}
    \label{lemma:positive_exposure}
    Under Assumption \ref{assump:exposure_overlap}, there exists a universal constant $\delta>0$ such that each content item in the consideration set has at least $\delta$ probability to be exposed: 
    \[
    \bbP\bp{k_i^*=k\mid V_i, \vv{C_i}, \vv{W_i}; s_0, s_1}\geq \delta, \quad \forall ~ (V_i, \vv{C_i}, \vv{W_i}, k).
    \]
\end{lemma}

\begin{proof}
We have
\begin{align*}
    \bbP\bp{k_i^*=k\mid V_i, \vv{C_i}, \vv{W_i}; s_0, s_1} &= \frac{e^{s_0(V_i, C_{i,k}) +  W_{i,k}\cdot s_1(V_i, C_{i,k})}}{\sum_{k'=1}^K e^{s_0(V_i, C_{i,k'}) +  W_{i,k'} \cdot s_1(V_i, C_{i,k'})}} \geq \frac{e^{-2C}}{\sum_{k'=1}^K e^{2C}} = e^{-2C-2\log(K)C},
\end{align*}
where the inequality is by Assumption \ref{assump:exposure_overlap}. Set $\delta=e^{-2C-2\log(K)C}$, concluding the proof.    
\end{proof}

\subsection{Score Identification}
\label{appendix:identification_proof}
We prove Proposition \ref{prop:identification} and show that $(s_0, s_1, z)$ can be nonparametrically identified up to a location normalization $s_0(V_i, C_{i,1})\equiv 0$.  Under Assumption \ref{assump:exposure_overlap}, there exists an $\delta>0$ such that $p_k(V_i, \vv{C}_i, \vv{W}_i;s_0,s_1)\geq \delta$ for all $V_i, \vv{C}_i, \vv{W}_i$ by Lemma \ref{lemma:positive_exposure}.

\paragraph{Identifiability of $(s_0, s_1)$.}
First set $\vv{W}_i=\zero$, which happens with probability of $(1-q)^K$ by the creator randomization, we  have $s_0$ be uniquely identified as:
\begin{equation}
    s_0(V_i, C_{i,k}) = \log\frac{p_k(V_i, \vv{C}_i, \zero; s_0, s_1)}{p_1(V_i, \vv{C}_i, \zero; s_0, s_1)},
\end{equation}
with the normalization $s_0(V_i, C_{i,1})\equiv 0$.
Subsequently, 
\begin{align*}
    s_1(V_i, C_{i,1}) = \log\bc{\bp{\sum_{k=2}^K e^{s_0(V_i, C_{i,k})}}\cdot\frac{p_1(V_i, \vv{C}_i, (1,\zero); s_0, s_1)}{1-p_1(V_i, \vv{C}_i, (1,\zero); s_0, s_1)}},
\end{align*}
% \textcolor{red}{is it $\log\bc{\bp{\sum_{k=2}^K e^{s_0(V_i, C_{i,k})}}\cdot \frac{p_1(V_i, \vv{C}_i, (1,\zero); s_0, s_1)}{1-p_1(V_i, \vv{C}_i, (1,\zero); s_0, s_1)}}$}
where $(1,\zero)$ denotes $W_{i,1}=1$ and $W_{i,2:K}=\zero$, which happens with probability $q(1-q)^{K-1}$ by experimental design.
Finally, set $\vv{W}_i=\one$, which happens with probability $q^K$. We have $s_1$ uniquely identified as:
\begin{equation}
    s_1(V_i, C_{i,k}) = \log\frac{p_k(V_i, \vv{C}_i, \one; s_0, s_1)}{p_1(V_i, \vv{C}_i, \one; s_0, s_1)} - s_0(V_i, C_{i,k}) + s_1(V_i, C_{i,1}).
\end{equation}

\paragraph{Identifiability of $z$.} For any $(V_i, \vv{C}_i)$, $k_i^*=k$ happens with probability at least $\delta$ by Lemma \ref{lemma:positive_exposure}, we thus  have $z(V_i, C_{i,k}) =\bbE[Y_i\mid k_i^*=k, V_i, \vv{C}_i]$ uniquely identified from the data distribution.

\subsection{Score Convergence}
\label{appendix:score_convergence}
\begin{assumption}
\label{assump:smoothness}
    (i) Both viewer covariates and content covariates have compact, connected support, with each entry taking values in $[-1,1]$. (ii) As functions of a viewer-content pair $(v,c)$, $s_0(\cdot,\cdot), s_1(\cdot,\cdot), z(\cdot,\cdot)\in \calW^{p,\infty}([-1,1]^{d_V+d_C})$, where for positive integers $p$ and $q$, define the H$\ddot{o}$lder ball $\calW^{p,\infty}([-1,1]^q)$ of function $h:\mathbb{R}^q\rightarrow \mathbb{R}$ with smoothness $p\in \mathbb{N}_+$ as
    \begin{equation*}
        \calW^{p,\infty}([-1,1]^q) := \bc{h:\max_{r,|r|\leq p} \underset{v\in[-1,1]^q}{\mathrm{ess~sup}} |D^rh(v)|\leq 1},
    \end{equation*}
    where $r=(r_1,\dots, r_q)$, $|r|=r_1+\dots + r_q$ and $D^r h$ is the weak derivative.
\end{assumption}
We prove Proposition \ref{prop:convergence} by verifying conditions of Theorem 1 in \cite{farrell2020deep}. We first verify Assumption 1 in \cite{farrell2020deep}.
\begin{itemize}
    \item The nuisance parameters $(s_0, s_1, z)$ are nonparametrically identified by Proposition \ref{prop:identification} and uniformly bounded by the model assumptions.
    \item Appendix \ref{appendix:gradient_of_ell} shows that the loss function has bounded gradients with bounded $z$ and outcome $Y$. Therefore it satisfies the Lipschitz condition.
    \item Write $\theta := (s_0, s_1, z)$ and $\tilde \theta:=(\tilde s_0, \tilde s_1, \tilde z)$ for any nuisance candidate $\tilde s_0, \tilde s_1, \tilde z$. By Taylor expansion, for some $\check\theta$ between $\tilde\theta$ and $\theta$, we have
    \begin{align*}
        &\bbE\bb{\ell(V_i, \vv{C}_i, \vv{W}_i, k_i^*, Y_i;\tilde \theta)\mid V_i, \vv{C}_i} -  \bbE\bb{\ell(V_i, \vv{C}_i, \vv{W}_i, k_i^*, Y_i; \theta)\mid V_i, \vv{C}_i} \\
        = &
         \nabla_\theta  \bbE\bb{\ell(V_i, \vv{C}_i, \vv{W}_i, k_i^*, Y_i; \theta)\mid V_i, \vv{C}_i}\|\tilde\theta(V_i, \vv{C}_i)-\theta(V_i, \vv{C}_i)\|_2 \\
         &\quad + \frac{1}{2}\nabla^2_\theta  \bbE\bb{\ell(V_i, \vv{C}_i, \vv{W}_i, k_i^*, Y_i; \check\theta)\mid V_i, \vv{C}_i}\|\tilde\theta(V_i, \vv{C}_i)-\theta(V_i, \vv{C}_i)\|^2_2\\
        = & \frac{1}{2}\nabla^2_\theta  \bbE\bb{\ell(V_i, \vv{C}_i, \vv{W}_i, k_i^*, Y_i; \check\theta)\mid V_i, \vv{C}_i}\|\tilde\theta(V_i, \vv{C}_i)-\theta(V_i, \vv{C}_i)\|^2_2,
    \end{align*}
    where the last equality is by the first-order optimality of $\theta$ in $\ell$. Lemma \ref{lemma:invertible_H} shows that the Hessian $\nabla^2_\theta  \bbE\bb{\ell(V_i, \vv{C}_i, \vv{W}_i, k_i^*, Y_i; \check\theta)\mid V_i, \vv{C}_i}$ is universally invertible with bounded inverse. Therefore, there exists a $c_1>0$ such that:
    \[
    \bbE\bb{\ell(V_i, \vv{C}_i, \vv{W}_i, k_i^*, Y_i;\tilde \theta)\mid V_i, \vv{C}_i} -  \bbE\bb{\ell(V_i, \vv{C}_i, \vv{W}_i, k_i^*, Y_i; \theta)\mid V_i, \vv{C}_i}\geq c_1 \|\tilde\theta(V_i, \vv{C}_i)-\theta(V_i, \vv{C}_i)\|^2_2.
    \]
    On the other hand, Appendix \ref{appendix:hessian} writes out the explicit form of the Hessian, which has bounded entries, implying bounded eigenvalues. Therefore, there exists a $c_2>0$ such that:
    \[
    \bbE\bb{\ell(V_i, \vv{C}_i, \vv{W}_i, k_i^*, Y_i;\tilde \theta)\mid V_i, \vv{C}_i} -  \bbE\bb{\ell(V_i, \vv{C}_i, \vv{W}_i, k_i^*, Y_i; \theta)\mid V_i, \vv{C}_i}\leq c_2 \|\tilde\theta(V_i, \vv{C}_i)-\theta(V_i, \vv{C}_i)\|^2_2.
    \]
\end{itemize}
We next verify Assumption 2 in \cite{farrell2020deep}.  The random variable are bounded with compact, connected support by our assumptions. The nuisance parameters are also uniformly bounded. Together with smoothness Assumption \ref{assump:smoothness}, we have Assumption 2 in \cite{farrell2020deep} hold. Then we apply Theorem 1 in \cite{farrell2020deep} and yield Proposition \ref{prop:convergence}.

\newpage

\section{Supporting Results for the Debiased Estimator}

\subsection{Invertible Hessian}
\label{appendix:invertible_H}
\begin{lemma}
\label{lemma:invertible_H}
    Suppose Assumptions \ref{assump:exposure_overlap} hold.  Suppose that the estimated scores are bounded: $\|\hat{s}_0\|_{\infty}\leq C$ and   $\|\hat{s}_1\|_{\infty}\leq C$ for the $C$ in Assumption \ref{assump:exposure_overlap}. Under our modeling framework, the expected hessian $H = \bbE\bb{\nabla^2\ell|V,\vv{C}}$ is universally invertible with bounded inverse.
\end{lemma}
\begin{proof}
    Note that $H$ is a two-block diagonal matrix with the first diagonal block being $H_1:=\bbE\bb{\nabla^2\ell_1|V,\vv{C}}$ and the second diagonal block being $H_2:=\bbE\bb{\nabla^2\ell_2|V,\vv{C}}$. It suffices to show that both $H_1$ and $H_2$ are universally invertible with bounded inverses.  

    \paragraph{Regularity of $H_1$.} Consider the sample $(V,\vv{C},\vv{W})$ and estimated nuisance $\hat{s}_0,\hat{s}_1$. For notation convenience, we use $p_k$ to represent $P(k^*=k; V,\vv{C}, \vv{W}, \hat{s}_0,\hat{s}_1)$. Define the vector
    \begin{align*}
        \beta = \begin{pmatrix}
        p_2 \\
        \vdots \\
        p_K\\
        W_1p_1\\
        \vdots\\
         W_{k}p_K
        \end{pmatrix} \in \bbR^{2K-1},
    \end{align*}
    the matrix $A\in \bbR^{(2K-1)\times (2K-1)}$ with its $(K,K)$-th entry being $\one[W_1=1]p_1$ and others zero and matrix
    \begin{equation*}
        B=\begin{pmatrix}
            \diag{(\sqrt{p_2},\dots, \sqrt{p_K})}\\
            0,\dots, 0,\\
            \diag{(\sqrt{p_2}W_2,\dots, \sqrt{p_K}W_{k})}
        \end{pmatrix}\in \bbR^{(2K-1)\times (2K-1)},
    \end{equation*}
    We thus have 
   \begin{align*}
       H_1 = \bbE\bb{
      A+ BB^T - \beta\beta^T 
       \mid V, \vv{C}
       }.
   \end{align*}
   We now show that $H_1$ is positive definite. For  
any vector $x=(x_1,\dots, x_{2K-1})$ and any treatment assignment $\vv{W}$, we have
\begin{align}
   & x^T(A+BB^T - \beta\beta^T)x  \label{eq:conditional_H}
    \\
    &= p_1W_{1}x_{K}^2 + \sum_{j=1}^{K-1} p_{j+1}(x_j + W_{j+1}x_{K+j})^2   - \bp{p_1W_{1}x_{K}+\sum_{j=1}^{K-1}
    p_{j+1}(x_j + W_{j+1}x_{K+j})
    }^2\nonumber\\
    &=\bp{p_1W_{1}x_{K}^2 + \sum_{j=1}^{K-1} p_{j+1}(x_j + W_{j+1}x_{K+j})^2}\bp{\sum_{j=1}^Kp_j} \nonumber  \\
   &\quad\quad  - \bp{p_1W_{1}x_{K}+\sum_{j=1}^{K-1}
    p_{j+1}(x_j + W_{j+1}x_{K+j})
    }^2 
    \stackrel{(i)}{\geq} 0, \nonumber
\end{align}
where (i) is by Cauchy-Swartzh inequality. That is, for any treatment assignment, the Hessian is positive semi-definite, and thus $H_1$ (the expected Hessian) is positive semi-definite. 

We next consider specific assignments of $\vv{W}$ to lower bound the smallest eigenvalue of $H_1$. Also note that under the assumption of the boundedness of score functions, there exists a constant $\delta>0$, such that for any $(U,\vv{V}, \vv{W}, k^*)$, we have 
    $
    P(k^*\mid U,\vv{V},\vv{W},\hat{s}_0, \hat{s}_1)\geq \delta.
    $.
\begin{description}
    \item[Consider $\vv{W}^{(a)}=(0, \dots, 0).$] Then we have  Equation~\eqref{eq:conditional_H} instantiate as
    \begin{align}
          x^T&(A+BB^T- \beta\beta^T)x  = \sum_{j=1}^{K-1} p^{(a)}_{j+1}x_j^2   - \bp{\sum_{j=1}^{K-1}
    p^{(a)}_{j+1}x_j 
    }^2\nonumber\\
   & = p_1^{(a)} \sum_{j=1}^{K-1} p^{(a)}_{j+1}x_j^2 + \bp{\sum_{j=1}^{K-1} p^{(a)}_{j+1}}\sum_{j=1}^{K-1} p^{(a)}_{j+1}x_j^2   - \bp{\sum_{j=1}^{K-1}
    p^{(a)}_{j+1}x_j
    }^2\geq p_1 \sum_{j=1}^{K-1} p^{(a)}_{j+1}x_j^2 \geq \delta^2 \sum_{j=1}^{K-1}x_j^2, \label{eq:no_item}
    \end{align}
    where the last inequality is due to Cauchy-Schwartz inequality.
    \item[Consider $\vv{W}^{(b)}=(0, 1, \dots, 1).$] Then we have  Equation~\eqref{eq:conditional_H} instantiate as
    \begin{align}
        x^T(A+&BB^T - \beta\beta^T)x  = \sum_{j=1}^{K-1} p^{(b)}_{j+1}(x_j + x_{K+j})^2   - \bp{\sum_{j=1}^{K-1}
    p^{(b)}_{j+1}(x_j + x_{K+j})
    }^2\nonumber\\
    &= p^{(b)}_1 \sum_{j=1}^{K-1} p^{(b)}_{j+1}(x_j + x_{K+j})^2 + \bp{\sum_{j=1}^{K-1} p^{(b)}_{j+1}}\sum_{j=1}^{K-1} p^{(b)}_{j+1}(x_j + x_{K+j})^2   - \bp{\sum_{j=1}^{K-1}
    p^{(b)}_{j+1}(x_j + x_{K+j})
    }^2\nonumber\\
    &\stackrel{(i)}{\geq} p^{(b)}_1 \sum_{j=1}^{K-1} p^{(b)}_{j+1}(x_j + x_{K+j})^2  \geq  \delta^2 \sum_{j=1}^{K-1} (x_j + x_{K+j})^2  = \delta^2 \sum_{j=1}^{K-1}\bp{x_j^2+x_{K+j}^2+2x_jx_{K+j}}\label{eq:item_1}\\
    &\stackrel{(ii)}{\geq} \delta^2 \sum_{j=1}^{K-1} (x_j^2+x_{K+j}^2-\frac{2}{3}x_{K+j}^2-\frac{3}{2}x_j^2) = \delta^2 \sum_{j=1}^{K-1} \bp{-\frac{1}{2}x_j^2+\frac{1}{3}x_{K+j}^2 }, \nonumber
    \end{align}
    where both (i) and (ii) are due to Cauchy-Schwartz inequality.

    \item[Consider $\vv{W}^{(c)}=(1,0,\dots,0).$] We have Equation~\eqref{eq:conditional_H} instantiate as
    \begin{align}
    \label{eq:item_2}
         x^T(A+BB^T - \beta\beta^T)x&= p^{(c)}_1x_{K}^2 + \sum_{j=1}^{K-1} p^{(c)}_{j+1}x_j^2   - \bp{p^{(c)}_1x_{K}+\sum_{j=1}^{K-1}
    p^{(c)}_{j+1}x_j
    }^2
    \end{align}
\end{description}
Note that under creator-side randomization,  each $\vv{W}$ has  assignment probability bounded away from zero, and we denote this lower bound as $\eta>0$. Also recall that each item exposure probability is lower bounded by $\delta>0$, with $K\delta\leq 1$.
    Combining Equation~\eqref{eq:no_item}, \eqref{eq:item_1}, \eqref{eq:item_2}, with the fact that the Hessian for any $\vv{W}$ is semi-definite, we have
    \begin{align}
       & x^TH_1x = \bbE[x^T(A+BB^T-\beta\beta^T)\mid V, \vv{C}] 
       \nonumber
       \\
       &\geq P\bp{\vv{W}=\vv{W}^{(a)}} \delta^2\sum_{j=1}^{K-1}x_j^2
         +  P\bp{\vv{W}=\vv{W}^{(b)}} \delta^2\sum_{j=1}^{K-1} \bc{-\frac{1}{2}x_j^2+\frac{1}{3}x_{K+j}^2} \nonumber \\
       &  \quad +  P\bp{\vv{W}=\vv{W}^{(c)}} \bc{p_1^{(c)}x_K^2+ \sum_{j=1}^{K} p_{j+1}^{(c)} x_j^2   - \bp{p_1^{(c)} x_{K}+\sum_{j=1}^{K-1}p_{j+1}^{(c)} x_j
    }^2} 
      \nonumber  \\
       &\geq \eta \delta^2\sum_{j=1}^{K-1}x_j^2
         + \frac{\eta\delta^2}{2}\sum_{j=1}^{K-1} \bc{-\frac{1}{2}x_j^2+\frac{1}{3}x_{K+j}^2}  + \frac{\eta \delta^2}{2(1-(K-1)\delta)^2} \bc{p_1^{(c)}x_K^2+ \sum_{j=1}^{K} p_{j+1}^{(c)} x_j^2   - \bp{p_1^{(c)} x_{K}+\sum_{j=1}^{K-1}p_{j+1}^{(c)} x_j
    }^2} \nonumber \\
    &=   \frac{\eta\delta^2}{2}  \sum_{j=1}^{K-1}\bc{\frac{1}{2}x_j^2+\frac{1}{3}x_{K+j}^2}
    +
I,\label{eq:regularity_H_1_a}
\end{align}
where $I=\frac{\eta\delta^2}{2(1-(K-1)\delta)^2} \bc{ p_1^{(c)} x_{K}^2 + \sum_{j=1}^{K-1} (p_{j+1}^{(c)}+(1-(K-1)\delta)^2) x_j^2   - \bp{p_1^{(c)} x_{K}+\sum_{j=1}^{K-1}p_{j+1}^{(c)} x_j
    }^2}$. We next lower bound term I.
\begin{align}
   I &\geq 
\frac{\eta\delta^2}{2(1-(K-1)\delta)^2} \bc{ p_1^{(c)} x_{K}^2 + \sum_{j=1}^{K-1} (1+p_1^{(c)})p_{j+1}^{(c)} x_j^2   - \bp{p_1^{(c)} x_{K}+\sum_{j=1}^{K-1}p_{j+1}^{(c)} x_j
    }^2} \nonumber \\
    &= \frac{\eta\delta^2}{2(1-(K-1)\delta)^2} \bc{\bp{p_1^{(c)}  x_{K}^2 + \sum_{j=1}^{K-1} (1+p_1^{(c)})p_{j+1}^{(c)} x_j^2}\bp{
    p_1^{(c)} +\sum_{j=1}^{K-1}\frac{p_{j+1}^{(c)}}{1+p^{(c)}_1} + \frac{p_1^{(c)}-(p_1^{(c)})^2}{p_1^{(c)}+1}
    }   - \bp{p_1^{(c)} x_{K}+\sum_{j=1}^{K-1}p_{j+1}^{(c)} x_j
    }^2} \nonumber \\
    &\stackrel{(i)}{\geq}  \frac{\eta\delta^2}{2(1-(K-1)\delta)^2} \bp{p_1^{(c)}  x_{K}^2 + \sum_{j=1}^{K-1} (1+p_1^{(c)})p^{(c)}_{j+1} x_j^2}\bp{
     \frac{p_1^{(c)}-(p_1^{(c)})^2}{p_1^{(c)}+1}
    } \nonumber  \\
    &\geq  \frac{\eta\delta^2}{2(1-(K-1)\delta)^2} \bp{p^{(c)}_1  x_{K}^2 + \sum_{j=1}^{K-1} (1+p^{(c)}_1)p^{(c)}_{j+1} x_j^2}\frac{p^{(c)}_1(1-p_1^{(c)})}{2}
    \nonumber  \\
    &\geq  \frac{\eta\delta^2(p^{(c)}_1)^2(1-p_1^{(c)})}{4(1-(K-1)\delta)^2}   x_{K}^2  \geq \frac{\eta\delta^5}{4(1-(K-1)\delta)^2}   x_{K}^2 , \label{eq:regularity_H_1_b}
    \end{align}
    where (i) is by Cauchy-Schwarz. Putting Equation~\eqref{eq:regularity_H_1_a} and \eqref{eq:regularity_H_1_b} together, for any $x=(x_1,\dots,x_{2K-1})$, we have
    \begin{equation}
        x^TH_1x \geq \frac{\eta\delta^2}{2}  \sum_{j=1}^{K-1}\bc{\frac{1}{2}x_j^2+\frac{1}{3}x_{K+j}^2} + \frac{\eta\delta^5}{4(1-(K-1)\delta)^2}   x_{K}^2.
    \end{equation}
    Therefore, $H_1$ has the smallest eigenvalue of   greater than or equal to $\frac{\eta\delta^2}{2}\min\bp{\frac{1}{3},
    \frac{\delta^3}{2(1-(K-1)\delta)^2}
    }$ and thus is invertible with bounded inverse.

    \paragraph{Regularity of $H_2$.} We have $H_2$ is a diagonal matrix with its $k$-th diagonal entry being
    \begin{align*}
        H_2(k,k)=\bbE\bb{\bbE\bb{P(k^*=k \mid U,\vv{V},\vv{W},\hat{s}_0, \hat{s}_1)}\mid U,\vv{V}}\geq \delta.
    \end{align*}
    As a result, $H_2\geq \delta \cdot I$ and thus is invertible with bounded inverse.
\end{proof}

\subsection{Universal Neyman Orthogonality}
\label{appendix:universal_orthogonality}

We show that the debiased estimate $\psi$, defined in Equation~\eqref{eq:debiase_estimate}, satisfies   the universal Neyman orthogonality \citep{chernozhukov2019semi,foster2023orthogonal}. This property means that the nuisance estimation error only has a second order effect on the debiased estimate -- a key property for achieving the asymptotic normality of the debiased estimator.

\begin{proposition}[Universal Orthogonality]
    The  debiased estimator  $\psi$ defined in Equation~\eqref{eq:debiase_estimate} is   universally orthogonal with respect to the nuisances  in the sense that, for any nuisance components $(\tilde{s}_0, \tilde{s}_1, \tilde{z}, \tilde{H})$, 
    \[
    \bbE[\nabla\psi(V,\vv{C}, \vv{W}, k^*, Y;\tilde{s}_0 = s_0, \tilde{s}_1 =s_1, \tilde{z}=z, \tilde{H}=H )\mid V, \vv{C}]=0,
    \]
    where   $(V,\vv{C}, \vv{W}, k^*, Y)$ is sampled  from the creator-side randomization experiment, and $\nabla\psi$ is the gradient with respect to the nuisances.

\end{proposition}

\begin{proof}
    Recall that 
    \begin{align}
\psi(V, \vv{C}, \vv{W}, k^*, Y; \tilde{s}_0, \tilde{s}_1,&\tilde{z}, \tilde{H} ) =~\mu(V,\vv{C}; \tilde{s}_0, \tilde{s}_1,\tilde{z}) \nonumber \\
&- \nabla\mu(V, \vv{C}; \tilde{s}_0, \tilde{s}_1,\tilde{z})^T \tilde{H}(V,\vv{C};\tilde{s}_0, \tilde{s}_1,\tilde{z})^{-1} \nabla\ell(V, \vv{C}, \vv{W}, k^*, Y; \tilde{s}_0, \hat{s}_1,\tilde{z}),\nonumber
\end{align}
Let $\tilde{h} = (\tilde{s}_0, \tilde{s}_1, \tilde{z}, \tilde{H}^{-1})$ with the ground truth $h = (s_0, s_1, z, H^{-1})$. It suffices to show that 
 \[
    \bbE[\nabla_{\tilde{h}}\psi(V,\vv{C}, \vv{W}, k^*, Y; \tilde{h}=h)\mid V, \vv{C}]=0.
\]
We have
\begin{align*}
    \frac{\partial \psi(\tilde{h}=h)}{\partial (\tilde{s}_0, \tilde{s}_1, \tilde{z})} = &\nabla\mu(V,\vv{C}; {s}_0, {s}_1,{z})   -\nabla^2 \mu(V,\vv{C}; {s}_0, {s}_1,{z}) H(V,\vv{C};s_0, s_1,z)^{-1} \nabla\ell(V, \vv{C}, \vv{W}, k^*, Y; s_0, s_1,z)\\
    &\quad- \nabla \mu(V,\vv{C}; {s}_0, {s}_1, {z}) H(V,\vv{C};s_0, s_1,z)^{-1} \nabla^2\ell(V, \vv{C}, \vv{W}, k^*, Y; s_0, s_1,z).
\end{align*}
 Then taking the expectation with respect to $(\vv{W},k^*,Y)$, we have
\begin{align*}
\bbE\bb{\frac{\partial \psi(\tilde{h}=h)}{\partial (\tilde{s}_0, \tilde{s}_1, \tilde{z})}\mid V,\vv{C}}& = \nabla\mu(V,\vv{C}; {s}_0, {s}_1,{z})\bp{I-H(V,\vv{C};s_0, s_1,z)^{-1} \bbE\bb{ \nabla^2\ell(V, \vv{C}, \vv{W}, k^*, Y; s_0, s_1,z)\mid  V,\vv{C}}}\\
&\quad\quad -\nabla^2 \mu(V,\vv{C}; {s}_0, {s}_1,{z}) H(V,\vv{C};s_0, s_1,z)^{-1} \bbE\bb{ 
\nabla\ell(V, \vv{C}, \vv{W}, k^*, Y; s_0, s_1,z)
\mid V, \vv{C}}\\
&\stackrel{(i)}{=} \nabla\mu(V,\vv{C}; {s}_0, {s}_1,{z})\bp{I-H(V,\vv{C};s_0, s_1,z)^{-1}H(V,\vv{C};s_0, s_1,z)}=0,
\end{align*}
where (i) is because the ground truth $(s_0,s_1,z)$ satisfies the first-order optimality condition of loss function $\ell$. 
Similarly, we have
\begin{align*}
    \bbE\bb{ \frac{\partial \psi(\tilde{h}=h)}{\partial (\tilde{H}^{-1})}\mid V,\vv{C} } = -\bbE\bb{\nabla\ell(V, \vv{C}, \vv{W}, k^*, Y; s_0, s_1,z)\mid V,\vv{C}}\nabla\mu(V,\vv{C}; {s}_0, {s}_1,{z})^T = 0.
\end{align*}
\end{proof}

\newpage
\section{Proof of the Asymptotics of the Debiased Estimator}
\label{appendix:clt}
We prove the below theorem to show the asymptotic regularity of the debiased estimator.
\clt* 

Define the oracle estimator $\tilde{\tau}^{DB}_n$ as
\begin{equation}
    \tilde{\tau}^{DB}_n = \frac{1}{n} \sum_{i=1}^n  \psi\bp{V_i, \vv{C_i}, \vv{W_i}, k^*_i, Y_i; s_0, s_1,z, H},
\end{equation}
with that 
\begin{align}
\psi(V_i, \vv{C_i}, \vv{W_i}, k^*_i, Y_i; s_0, s_1,&z, H ) =~\mu(V_i,\vv{C_i}; s_0, s_1, z) \\
&- \nabla\mu(V_i, \vv{C_i}; s_0, s_1, z)^T H(V_i,\vv{C_i};s_0, s_1, z)^{-1} \nabla\ell(V_i, \vv{C_i}, \vv{W_i}, k^*_i, Y_i; s_0, s_1, z).\nonumber
\end{align}
Note that 
\begin{equation}
    \bbE[\psi(V_i, \vv{C_i}, \vv{W_i}, k^*_i, Y_i; s_0, s_1,z, H )\mid V_i, \vv{C_i}] = \mu\bp{V_i, \vv{C_i}; s_0, s_1,z},
\end{equation}
by the first order condition of $(s_0, s_1,z)$ in $\ell$ when conditioning on $(V_i, \vv{C}_i, \vv{W}_i)$. 

Our proof follows three main steps: 
\begin{enumerate}
    \item[(i)] Show $\hat{\tau}_n^{DB}$ and  $\tilde{\tau}_n^{DB}$ have similar asymptotic behavior.
    \item[(ii)] Show asymptotic normality of $\tilde{\tau}_n^{DB}$.
    \item[(iii)] Show asymptotically consistent variance estimator $\widehat{V}_n$.
\end{enumerate}

\subsection{\texorpdfstring{Step I: Show $\hat{\tau}_n^{DB}$ and $\tilde{\tau}_n^{DB}$ have similar asymptotic behavior.}{Similar asymptotics}}
This part's proof is largely motivated by the proof pattern for Theorem 3.1 in \cite{chernozhukov2018double}, except that we need to additionally deal with the correlation among samples due to the shared items in their consideration sets. 
For notation convenience, define $\theta_0=(s_0, s_1,z, H)$ and $\hat{\theta}_{0} = \bp{\hat{s}_0, \hat{s}_1,\hat{z}, \hat{H}}.$ Write $Z_i:=(V_i, \vv{C_i}, \vv{W_i}, k_i^*, Y_i)$. 
 Note that 
\begin{align}
     \bbE\bb{\psi\bp{Z_i; \theta_0}\mid V_i, \vv{C}_i, \vv{W}_i}\stackrel{(i)}{=} \bbE\bb{\mu\bp{V_i, \vv{C_i}; s_0, s_1,z}\mid V_i, \vv{C}_i, \vv{W}_i} = \mu\bp{V_i, \vv{C_i}; s_0, s_1,z},
\end{align}
where $(i)$ is by the first-order optimality of $(s_0, s_1, z)$ in $\ell$. 
We have
\begin{align}
    \sqrt{n}\ba{\hat{\tau}_n^{DB} -\tilde{\tau}_n^{DB}} 
    \leq I_{1} + I_{2},
\end{align}
where 
\begin{align*}
    I_{1}: &= \bn{\frac{1}{\sqrt{n}} \sum_{i}\bp{\psi(Z_i;\hat{\theta}_0) - \bbE[\psi(Z_i;\hat{\theta}_0)]} - \frac{1}{\sqrt{n}} \sum_{i}\bp{\psi(Z_i;\theta_{0}) - \bbE[\psi(Z_i;\theta_{0})]}},\\
     I_{2}: &= \bn{\sqrt{n} \bp{ \bbE[\psi(Z_i;\hat{\theta}_{0})] - \bbE[\psi(Z_i;\theta_{0})]} }.
\end{align*}
We now bound $I_{1}$ and $I_{2}$ respectively. 

\emph{Bounding $I_{1}$.}
Write $Q_i :=\bp{\psi(Z_i;\hat{\theta}_{0}) - \bbE[\psi(Z_i;\hat{\theta}_{0})]} -\bp{\psi(Z_i;\theta_{0}) - \bbE[\psi(Z_i;\theta_{0})]}$.
We have
\begin{align}
    \bbE[I_{1}^2] = & \frac{1}{n}\bbE\bb{\bp{\sum_iQ_i}^2}=\frac{1}{n}\sum_i\bbE[Q_i^2] + \frac{1}{n}\sum_{i}\sum_{j\neq i} \one\bp{\vv{W}_i\cap \vv{W}_j\neq \emptyset}\bbE[Q_iQ_j]\nonumber\\
    =&\frac{O(a_n)}{n}\sum_i\bbE[Q_i^2] = O\bp{\frac{a_n}{n}}\sum_i\bbE\bb{\bc{\bp{\psi(Z_i;\hat{\theta}_{0}) - \bbE[\psi(Z_i;\hat{\theta}_{0})]} -\bp{\psi(Z_i;\theta_{0}) - \bbE[\psi(Z_i;\theta_{0})]}}^2}\\
    \leq & O\bp{\frac{a_n}{n}}
 \sum_i \bbE\bb{\bp{\psi(Z_i;\hat{\theta}_{0}) -\psi(Z_i;\theta_{0})} ^2} =O(a_n\epsilon_n^2).
\end{align}
Therefore, by Markov inequality, we have $I_1=O_p(a_n^{1/2}\epsilon_n)$. 

\emph{Bounding $I_{2}$.}
 Define the function 
\begin{equation}
    f(r):=\bbE[\psi(Z_i;\theta_0 + r(\hat{\theta}_{0}-\theta_0))] - \bbE[\psi(Z_i;\theta_0))],\quad r\in(0,1).
\end{equation}
By Taylor expansion, we have 
\begin{equation}
    f(1) = f(0) + f'(0) + f''(\tilde{r})/2,\quad \mbox{for some }\tilde{r}\in(0,1).
\end{equation}
Note that by the Neyman orthogonality (shown in \ref{appendix:universal_orthogonality}), we have $f'(0)= 0$.
With the bounded inverse of Hessian, we have
\begin{equation}
    \bbE[\|f''(\tilde{r})\|]\leq \sup_{r\in(0,1)}\|f''(r)\|= O(\epsilon_n^2).
\end{equation}
We have
\begin{align}
     I_{2,k}  =O_p \bp{\sqrt{n}\epsilon_n^2}.
\end{align}
Combining the bound for $I_1$ and $I_2$, we have
\begin{align}
   \sqrt{n}\bp{\hat{\tau}_n^{DB} - \tilde{\tau}_n^{DB}} =O_p\bp{a_n^{1/2} \epsilon_n +  n^{1/2}\epsilon_n^2 }=o_p(1), \label{eq:tilde_hat_tau}
\end{align}
with $a_n=O(n^{1/4})$ and $\epsilon_n=o(n^{-1/4})$.

\subsection{\texorpdfstring{Step II: Show asymptotic normality of $\tilde{\tau}_n^{DB}$.}{Asymptotic normality}
}\label{appendix:clt_tautilde}
Recall the data generating process: when a new viewer $V_i$ arrives, the back-end retrieval system firstly generates the consideration set $\vv{C}_i$. For content items that have shown in previous samples, the  treatment status remains unchanged. For content items that haven't appear, we sample the treatment status from i.i.d.~Bernoulli randomized trials. This procedure constructs the treatment collection $\vv{W}_i$. Then given $(V_i, \vv{C}_i, \vv{W}_i)$, the algorithm chooses  item $k_i^*$ to expose, yielding the viewer outcome  $Y_i$ and the observation tuple $Z_i:=(V_i, \vv{C}_i, \vv{W}_i, k_i^*, Y_i)$. We now apply the martingale theorem to analyze the asymptotic behavior of $\tilde{\tau}_n$. Denote the $\sigma$-field $\calF_i:=\sigma(Z_{1},\dots Z_i)$. We have that
\begin{align*}
    \bbE[\psi(Z_i; \theta_0)\mid \calF_{i-1}] = &\bbE\bb{
    \bbE\bb{\psi(Z_i; \theta_0)\mid V_i,\vv{C}_i, \vv{W_i} }\mid \calF_{i-1}
    } \\
    =&  \bbE\bb{
    \bbE\bb{\mu(V_i,\vv{C}_i; \theta_0)\mid V_i,\vv{C}_i, \vv{W_i} }\mid \calF_{i-1}
    } = \bbE\bb{\mu(V_i,\vv{C}_i; \theta_0)} = \tau.
\end{align*}
Therefore $\{\psi(Z_i; \theta_0) -\tau\}$ forms a martingale difference sequence with respect to filtration $\{\calF_i\}$.
We now apply the following result from \cite{hall2014martingale}. 
\begin{proposition}[Martingale Central Limit Theorem, Theorem 3.2, \cite{hall2014martingale}]
    Let $\{\xi_i\}$ be a martingale difference sequence with respect to filtration $\{\calF_i\}$, and let $\eta^2$ be an a.s.~finite random variable. Suppose that:
    \begin{align}
        &\max_i |\xi_i|\xrightarrow{p}0,\label{eq:mds_decay}\\
        &\sum_i \xi_i^2 \xrightarrow{p}\eta^2,\label{eq:mds_var}\\
        &\bbE[\max_i \xi_i^2] \text{ is bounded}.\label{eq:mds_bound}
    \end{align}
    Then $\sum_{i=1}^n X_i \xrightarrow{d}Z$ (stably), where the random variable $Z$ has characteristic function $\bbE[\exp(-\frac{1}{2}\eta^2t^2)]$.
\end{proposition}
Now let's verify the above conditions for $\sqrt{n}(\Tilde{\tau}-\tau)$. Define
\[
\xi_i = \frac{1}{\sqrt{n}} \bc{\psi(Z_i;\theta_0) - \tau}.
\]
We have $\bbE[\xi_i|\calF_{i-1}]=0$. By the bounded inverse of Hessian, we have $|\xi_i| = O(n^{-1/2})$, implying Equation~\eqref{eq:mds_decay} and \eqref{eq:mds_bound}. Section \ref{appendix:variance_convergence} shows that:
\[
\widetilde{V}_n:=\sum_i\xi_i^2 = \frac{1}{n}\sum_i (\psi_i(Z_i;\theta_0)-\tau)^2 \xrightarrow{p}\eta^2,
\]
 where $\eta^2:=\Var(\psi_i(Z_i;\theta_0)$ that is a finite number, verifying Equation~\eqref{eq:mds_var}.
We thus have
\begin{equation}
\label{eq:tilde_tau_tau}
    \sqrt{n}(\Tilde{\tau}_n^{DB}-\tau) =\sum_i \xi_i \xrightarrow{d} \calN(0,\eta^2),
\end{equation}
where $z$ is a Gaussian random variable $\calN(0, \eta^2)$ that has characteristic function $\bbE[\exp(-\frac{1}{2}\eta^2 t^2)]$. Connecting Equation~\eqref{eq:tilde_tau_tau} with Equation~\eqref{eq:tilde_hat_tau}, by Slutsky's theorem, we have
\begin{equation}
    \sqrt{n}\bp{\hat{\tau}_n^{DB} - \tau} = \sqrt{n}\bp{\hat{\tau}^{DB}_n - \tilde{\tau}^{DB}_n} + \sqrt{n}\bp{\tilde{\tau}^{DB}_n - \tau} \xrightarrow{d} z,
\end{equation}
yielding that $   \sqrt{n}(\hat{\tau}^{DB}_n - \tau)/\eta$ is asymptotically normal and converges to $\calN(0,1)$.

\subsection{\texorpdfstring{Step III: Variance convergence and connecting $\widetilde{V}_n$ to $\widehat{V}_n$.}{variance convergence}} 
\label{appendix:variance_convergence}
We start by showing that $\widetilde{V}_n$ converges in probability to $\eta^2:=\Var(\psi_i(Z_i;\theta_0)$. We first argue that $\bbE\bb{(\psi_i(Z_i;\theta_0)-\tau)^2}$ has the same value, denoted as $\eta^2$. Note that $\bbE\bb{\psi_i(Z_i;\theta_0)}  = \bbE\bb{\bbE\bb{\psi_i(Z_i;\theta_0)\mid \calF_{i-1}}} = \tau$. Therefore $\bbE\bb{(\psi_i(Z_i;\theta_0)-\tau)^2}$  represents the variance of $\psi_i(Z_i;\theta_0)$. Given $n$ samples, let $\vv{w}$ denote the treatment status of all items involved in the $n$ samples. We  have that each entry of $\bfw$ is i.i.d.~sampled from Bernoulli distribution with the treated probability being $q$, for which we abuse notation and denote as $\bfw\sim\calB(q)$. Therefore,
\begin{align*}
   \Var(\psi_i(Z_i;\theta_0)):&= \bbE\bb{(\psi_i(Z_i;\theta_0)-\tau)^2} =\bbE_{\bfw\sim \calB(q)}\bb{\bbE\bb{(\psi_i(Z_i;\theta_0)-\tau)^2\mid \bfw} }\\
   &=\bbE_{\vv{W}_i\sim\calB(q)}\bb{\bbE\bb{(\psi_i(Z_i;\theta_0)-\tau)^2\mid \vv{W}_i} }\\
   &=\bbE_{(V_i, \vv{C}_i),\vv{W}_i\sim\calB(q)}\bb{\bbE\bb{(\psi_i(Z_i;\theta_0)-\tau)^2\mid V_i, \vv{C}_i,\vv{W}_i} }.
\end{align*}
Recall that $\{(V_i, \vv{C}_i)\}_{i=1}^n$ are sampled i.i.d.~from i.i.d.~viewer queries, we thus obtain that $\Var(\psi_i(Z_i;\theta_0))$ across $i$ have the same value, denoted as $\eta^2$. Similarly, we have that for any  $m$-moment  $\bbE[(\psi_i(Z_i;\theta_0)-\tau)^m]$, they have the same bounded value across sample $i$, where the boundedness comes from  the boundedness of $\psi_i$.

We thus have $\bbE[\widetilde{V}_n]=\eta^2$. We proceed to show $\widetilde{V}_n \xrightarrow{p}\eta^2$. By Markov inequality, given any $\epsilon>0$, we have
\begin{align}
\label{eq:Vn_converge}
    P\bp{|\widetilde{V}_n - \eta^2|\geq \epsilon}\leq \frac{\Var(\widetilde{V}_n)}{\epsilon^2}.
\end{align}
Note that 
\begin{align*}
\Var\bp{\widetilde{V}_n} 
&= \frac{\sum_{i=1}^n\bbE\bb{(\psi_i(Z_i;\theta_0)-\tau)^4} }{n^2}
+\frac{\sum_{i\ne j}\one\bp{\vv{W}_i\cap \vv{W}_j\neq \emptyset} \Cov\bc{(\psi_i(Z_i;\theta_0)-\tau)^2,(\psi_j(Z_j;\theta_0)-\tau)^2} }{n^2}\\
&\stackrel{(i)}{=} O\bp{\frac{n}{n^2}} + \frac{na_n}{n^2}=O\bp{\frac{a_n}{n}}=o(n^{-3/4})=o(1),
\end{align*}
where $(i)$ uses the boundedness of $\bbE\bb{(\psi_i(Z_i;\theta_0)-\tau)^4}$ and the Cauchy-Schwartz inequality such that $ \Cov\bc{(\psi_i(Z_i;\theta_0)-\tau)^2,(\psi_j(Z_j;\theta_0)-\tau)^2}\leq \sqrt{\bbE\bb{(\psi_i(Z_i;\theta_0)-\tau)^4}\bbE\bb{(\psi_j(Z_j;\theta_0)-\tau)^4}}$. Continuing Equation~\eqref{eq:Vn_converge}, we have
\begin{equation*}
      P\bp{|\widetilde{V}_n - \eta^2|\geq \epsilon}\leq \frac{\Var(\widetilde{V}_n)}{\epsilon^2} = o(1),
\end{equation*}
implying that $\widetilde{V}_n \xrightarrow{p} \eta^2$. Therefore,  $ \sqrt{n}\bp{\hat{\tau}^{DB}_n - \tau}/\eta$ converges in distribution to  $\calN(0,1)$.

 We now show that $\widehat{V}_n$  also converges in probability to  $\eta^2$.  We have 
\begin{align*}
    |\widehat{V}_n - \widetilde{V}_n| &= \ba{\frac{1}{n}\sum_i\bp{\psi(Z_i;\hat{\theta}_0) - \psi(Z_i;\theta_0) - \hat{\tau}^{DB}_n +\tau}\bp{\psi(Z_i;\hat{\theta}_0) + \psi(Z_i;\theta_0)- \hat{\tau}^{DB}_n -\tau}}\\
    &\stackrel{(i)}{\leq}  O(1)\cdot \bc{ \frac{1}{n}\sum_i\ba{\psi(Z_i;\hat{\theta}_0) - \psi(Z_i;\theta_0)}+\frac{1}{n}\sum_i |\hat{\tau}^{DB}_n -\tau|}\\
    &= O_p(\|\hat{\theta}_0 - \theta_0\|_{L_2}) + O_p(n^{-1/2}) = O_p(\epsilon_n) + O_p(n^{-1/2})=o_p(1),
\end{align*}
where (i) is by the  boundedness of $\psi(Z_i;\hat{\theta_0})$ and $\psi(Z_i;\theta_0)$.

Therefore,
\[
 |\widehat{V}_n - \eta^2|\leq  |\widehat{V}_n - \widetilde{V}_n|+ |\eta^2-\widetilde{V}_n| =o_p(1).
\]
Combining the above with $\sqrt{n}(\hat{\tau}^{DB}-\tau)/\eta\xrightarrow{d}\calN(0,1)$, by Slutsky's theorem, we have $\sqrt{n}(\hat{\tau}^{DB}-\tau)/\sqrt{\widehat{V}_n}\xrightarrow{d}\calN(0,1)$, concluding the proof.

\newpage

\end{APPENDICES}

\end{document}